\documentclass[reqno]{amsart}

\usepackage{tikz}
\usetikzlibrary{shapes,arrows,patterns}
\usepackage{enumerate}
\usepackage{blkarray}
\usepackage{float}
\usepackage[a4paper,text={6.1in,8in},centering,includefoot]{geometry}
\normalsize
\usepackage{amsthm,dsfont,mathrsfs,amsmath,amssymb,amsfonts,euscript,mathtools}
\usepackage{graphicx}
\usepackage[colorlinks=true,citecolor=blue,linkcolor=red,urlcolor=blue]{hyperref}

\newtheorem{theorem}{Theorem}

\newtheorem{corollary}[theorem]{Corollary}
\newtheorem{lemma}[theorem]{Lemma}
\theoremstyle{remark}
\newtheorem{remark}[theorem]{Remark}

\theoremstyle{definition}
\newtheorem{definition}[theorem]{Definition}
\newtheorem{proposition}[theorem]{Proposition}

\begin{document}

\title[The occurrence of riddled basins and blowout bifurcations in a parametric nonlinear system]{The occurrence of riddled basins and blowout bifurcations in a parametric nonlinear system}%

\author{M. Rabiee}
\address{Department of Mathematics, Ferdowsi University of Mashhad, Mashhad, Iran.}
\email{maryam\_rabieefarahani@yahoo.com}
\author{F. H. Ghane}
\address{Department of Mathematics, Ferdowsi University of Mashhad, Mashhad, Iran.}
\email{ghane@math.um.ac.ir}
\author{M. Zaj}
\address{Department of Mathematics, Ferdowsi University of Mashhad, Mashhad, Iran.}
\email{zaj.marzie@yahoo.com}
\author{S. Karimi}
\address{Department of Mathematics, Ferdowsi University of Mashhad, Mashhad, Iran.}
\email{sohrabkarimi18@gmail.com}

 \thanks{$^*$Corresponding author}

\begin{abstract}
In this paper, a two parameters family $F_{\beta_1,\beta_2}$ of maps of the plane living two different subspaces invariant is studied.
We observe that, our model exhibits two chaotic attractors $A_i$, $i=0,1$, lying in these invariant subspaces
and identify the parameters at which $A_i$ has a locally riddled basin of attraction or becomes a chaotic saddle.
Then, the occurrence of riddled basin in the global sense is investigated in an open region of $\beta_1\beta_2$-plane. We semi-conjugate our system to a random walk model
and define a fractal boundary which separates the basins of attraction of the two
chaotic attractors, then we describe riddled basin in detail.
We show that the model undergos a sequence of bifurcations: ``a blowout
bifurcation", ``a bifurcation to normal repulsion" and ``a bifurcation by creating a new chaotic attractor
with an intermingled basin".
Numerical simulations are presented graphically to confirm the validity of our results.
\end{abstract}
\maketitle

\section{\textbf{Introduction}}
There has been a lot of recent interest in the global dynamics of systems with multiple attractors,
with the recognition that the structure of basins of attraction may be very complicated.
The terminology of multiple attractors, allows several attractors to coexist, has appeared in \cite{Bu, Ott, Da, D}.
It is very common for dynamical systems to have more than one attractor.
Among of them, we focus on systems having multiple attractors may present basins of attraction densely blended, a phenomenon called \emph{riddling}. This means that
for every initial condition in the basin of one attractor there are arbitrary near conditions which tend to
the basin of any of the other attractors.
Ott et al. \cite{OS} introduced non-linear dynamical systems with a simple symmetry contain riddled
basins. Also, the conditions of occurrence riddled basins are defined in Alexander et al. \cite{JI} and then generalized by Ashwin et al. \cite{PJ}.

Riddling basins arise in systems that possess chaotic
dynamics in a smooth invariant manifold of lower dimension than that of the full phase space.
This complexity appears at the transition point between strong and weak stability of the invariant subspace.
When riddled basins occur, differing initial conditions of the different replicates may asymptotic to different attractors.
Hence, the task of predicting what will be the final state of the system becomes difficult.
A detailed picture of multiple attractors with riddled basins is available through works of several authors \cite{JI, PJ, LH, MA, S}.
Applications of riddling in complex dynamical systems of physical and biological interest can be found: for instance, for a forced double-well Duffing oscillator \cite{OS1, OS, EJI},
ecological population models \cite{C, RS, KG},
learning dynamical systems \cite{NU}, coupled non-linear electronic circuits \cite{ABS,HC}, among others.

Take a nonlinear dynamical system possesses a smooth invariant manifold $N$.
Suppose that the restriction of the system to $N$ has an attractor $A$, so it is stable to perturbation within $N$.
Assume the behaviour of the system near $A$ is determined by combining the dynamics on $A$ and the dynamics transverse to $N$.
In the case that $A$ is chaotic, by the global nature of the effect of perturbations transverse
to $N$, considerable complexities in dynamics is observed \cite{PJ, ABS, C, RS}.
In this situation, the local dynamic stability of the chaotic attractor $A$ may be described in terms of normal Lyapunov exponents. When the largest normal
Lyapunov exponent is negative, there is a set of positive measure which is forward asymptotic to $A$ \cite{JI, PJ}.
In general, when we discuss the occurrence of the riddled basin for the attractor $A$, it is necessary to have a dense set of points with
zero Lebesgue measure in $A$ lying in the invariant subspace which are transversely unstable,
thus it is necessary that the attractor $A$ be chaotic.
In most chaotic attractors ergodic measures are not unique; for instance, they may exhibit dirac ergodic measures whose
support are periodic orbits. Each ergodic measure carries its own Lyapunov exponents, so the
stability in transverse directions can be considered independently for every ergodic measure supported
in that attractor. For example, two different periodic orbits in the attractor may have normal exponents of
different signs. If there exists a natural measure on $A$, then Lebesgue-almost all points
have corresponding normal exponents and manifolds, but there can still be a dense set in $A$
with the opposite behaviour \cite{PJ}.

When the full space contains two chaotic attractors lying in different invariant subspaces, the system presents a
complex fractal boundary between the initial conditions leading to each of the two attractors.
However, in a riddled basin, small variations in initial conditions induce a switch between the different asymptotic attractors but
fractal boundary causes to predict, from a given initial condition, what trajectory in phase-space the system will follow.

Ott et al. \cite{EJI} observed the occurrence of riddling basin for a certain nonlinear model of point particle motion subject to friction
and periodic forcing in a two-dimensional potential and tested this observation numerically.
However, they verified their results theoretically by calculations leading to a simple piecewise-linear model.

In this article, we examine the behavior of a two parameter family of planar systems $F_{\beta_1,\beta_2}$.
In our setting, each system $F_{\beta_1,\beta_2}$ exhibits two invariant subspaces $N_i$, $i=1,2$, having chaotic attractors $A_i$.
We demonstrate the emergence conditions of locally riddled basin, chaotic saddle and riddled basin in the global sense.
The blowout bifurcations of chaotic attractors in these invariant subspaces are illustrated.
We give a detailed analysis to occurrence of these phenomena.
It can be done by semi-conjugating the system to a random walk model and defining a fractal boundary which separates the basins of attraction of two chaotic attractors.
Using this approach, we investigate the system at riddled basin and blowout in detail.

The setting of this paper is
a family $\mathcal{F}$ of two parametric skew product maps of the form
\begin{align}\label{ss}
F_{\beta_1,\beta_2}: \mathbb{I}\times \mathbb{I} \to \mathbb{I}\times \mathbb{I}, \ F_{\beta_1,\beta_2}(x,y):=(f(x),g_{\beta_1,\beta_2}(x,y)),
\end{align}
where $\mathbb{I}$ is the unit interval $[0,1]$, $f$ is an expanding Markov map given by

\begin{align}\label{base}
f(x) =\left\{\begin{array}{cc}
2x & $ for$ \quad0\leq x\leq 1/2,\\
2x-1 & $ for$ \quad 1/2< x\leq1
\end{array}\right.
\end{align}
and
\begin{align}\label{fiber}
g_{\beta_1,\beta_2}(x,y) :=\left\{\begin{array}{cc}
g_{1,\beta_1}(y) & $ for$ \quad0\leq x\leq 1/2,\\
 g_{2,\beta_2}(y) & $ for$ \quad 1/2< x\leq1.
\end{array}\right.
\end{align}

We assume that the $C^2$ diffeomorphisms $g_{i,\beta_i}:\mathbb{I} \to \mathbb{I}$, $i=1, 2$, fulfill the following conditions:
\begin{enumerate}
  \item [(I1)] $g_{i,\beta_i}(0)=0, \quad g_{i,\beta_i}(1)=1, \quad g_{i,\beta_i}(\beta_i)=\beta_i, \quad \text{for} \ i=1, 2;$
  \item [(I2)]
   $g_{1,\beta_1}(y)<y \quad \text{for} \ y< \beta_1, \quad  g_{1,\beta_1}(y)>y \quad \text{for} \ y> \beta_1;$
    \item [(I3)]
  $  g_{2,\beta_2}(y)>y \quad \text{for} \ y< \beta_2, \quad  g_{2,\beta_2}(y)<y  \quad \text{for} \ y> \beta_2;$
  \end{enumerate}

In (\ref{ss}), the subspaces
\begin{equation}\label{subspaces}
 N_{0}:=\mathbb{I} \times \{0\}, \quad N_{1}:=\mathbb{I} \times \{1\}
\end{equation}
play the role of the $F_{\beta_1,\beta_2}$-invariant manifolds with chaotic dynamics inside, for each $\beta_1, \beta_2 \in (0,1)$.
Our objective is to study some types of bifurcations by varying the parameters $\beta_i$ and characterize the
different possible dynamics.

Particular example is given by
\begin{align}\label{fiber1}
g_{\beta_1,\beta_2}(x,y) :=\left\{\begin{array}{cc}
g_{1,\beta_1}(y)=y+y(1-y)(y-\beta_1) & $ for$ \quad0\leq x\leq 1/2,\\
 g_{2,\beta_2}(y)=y-y(1-y)(y-\beta_2) & $ for$ \quad 1/2< x\leq1.
\end{array}\right.
\end{align}

The parameters $\beta_i$, $i=1,2$, vary the transverse dynamics without changing the dynamics on the invariant
subspaces $N_0$ and $N_1$.
We will show that, for some values of $\beta_i$, the parametric family $F_{\beta_1,\beta_2}$ exhibits two attractors lying in invariant subspaces $N_i$ with a qualitative dynamics which depends on
the initial conditions.
Attractors in our model, for some parameters values, exhibit a complex attracting basin structure
that is riddled by holes. This phenomenon produces an unpredictability qualitatively greater than the traditional sensitive
dependence on initial conditions within a single chaotic attractor.
The dynamics of the system is described by two Lyapunov exponents. The first one is
 the parallel Lyapunov exponent which describes the evolution on the invariant subspaces and must be
positive for emergence of riddled basins. The second is the normal Lyapunov exponent that characterizes
evolution transverse to the subspaces \cite{PJ, C, RS}.

In this paper, we estimate the range of values of
parameters $\beta_i$, $i=1,2$, such that the attractors $A_i$ have a locally riddled basin, riddled basin in the global sense or becomes
a chaotic saddle and provide rigorous analysis for these complex behaviors. Also, we investigate that, when $\beta_1=\beta_2$, a new chaotic attractor is born
 and its basin is intermingled with the basins of both chaotic attractors $A_i$, $i=1,2$.
We show that by varying the parameters $\beta_i$, we are forced to undergo a sequence of bifurcations:
 a ``blowout bifurcation", a ``bifurcation to normal repulsion" and ``a bifurcation by creating a new chaotic attractor with an intermingled basin".

Here, we will show that by varying the parameters $\beta_i$, it is possible one of the chaotic sets in the invariant subspaces looses the stability when the parameters pass
 through critical values.
This happens if the normal Lyapunov exponent of one of the chaotic attractors crosses zero at these critical values.
If the normal Lyapunov exponent is negative, there is a set of positive measure which is forward asymptotic to that attractor. However, there can still
be an infinite set of trajectories in the neighborhood of the attractor that are repelled from it.
In particular, if the normal Lyapunov exponent is small and negative, the blowout bifurcation occurs and the riddling basin can be observed near
 the bifurcation point.

This paper is organized as follows. In Section 2, we describe precisely the notions and terminology used in this paper. In Section 3, we concentrate on studying the two parameters family $F_{\beta_1,\beta_2}$.
Using the results of \cite{PJ}, we investigate the occurrence of locally riddled basins and chaotic saddles for some values of parameters in an open region of $\beta_1\beta_2$-plane.
In Section 4, we introduce a random walk model which is semi-conjugate to the skew product system $F_{\beta_1,\beta_2}$. This allows us to define a fractal boundary between the initial conditions leading
to each of the two attractors $A_0$ and $A_1$. In Section 5, we demonstrate the emergence conditions of riddled basins in the global sense and indicate
 the parameters values for which the family undergos a sequence of bifurcations: a blowout bifurcation
and a bifurcation by creating a new chaotic attractor with intermingled basin.
 \section{\textbf{Terminology }}
In this section, we introduce the concepts and the notations which are basic in this paper.
\subsection{ Attractors and riddled basins}
Let $M$ be a compact connected smooth Riemannian
manifold and $m$ denotes the normalized Lebesgue measure. First we recall some classical
definitions related to attractors.

Let $F : M \to M$ be a continuous map and  $A\subset M$ is a compact $F$-invariant set ( i.e. $F(A) = A$). We say $A$  is \emph{transitive} if there exists $x\in A$ such that $\omega(x)=A$, where $\omega(x)$ is the set of limit points of the orbit $\{F^{n}(x)\}_{n\geq0}$.
 \emph{The basin of attraction} of $A$, we denote it by $\mathcal{B}(A)$, is the set of points whose $\omega$-limit set is
contained in $A$.
For non-empty $A$ the basin $\mathcal{B}(A)$ is always non-empty because it includes $A$. For $A$ to be an attractor, we require that $\mathcal{B}(A)$ is large in the appropriate sense.
The compact invariant set $A$ is called an \emph{asymptotically stable attractor} if it is Lyapunov stable and the
basin of attraction $\mathcal{B}(A)$ contains a neighbourhood of $A$.

Many variants for the definition of an attractor can be found the literature, see \cite{Mi} for a
discussion. In a weaker form \cite{Mi}, we say that a compact $F$-invariant set $A$ is an attractor in Milnor sense for $F$ if the basin of attraction $\mathcal{B}(A)$ has
positive Lebesgue measure.
To be more precise, we say that
A is a \emph{Milnor attractor} if $\mathcal{B}(A)$ has non-zero Lebesgue measure and there is no compact proper subset $A^{\prime}$ of $A$ whose basin coincides with $\mathcal{B}(A)$ up to a set of zero measure.
Melbourne introduced a stronger form of a Milnor attractor \cite{Mel}.
$A$ is called an \emph{essential attractor} if
\begin{equation*}
\lim_{\delta \to 0}\frac{m(B_\delta(A) \cap \mathcal{B}(A))}{m(B_\delta(A))}=1,
\end{equation*}
where $B_\delta(A)$ is a $\delta$-neighbourhood of $A$ in $M$.

Here, we deal with chaotic attractors.
A compact $F$-invariant set $A$ is a \emph{chaotic attractor} if $A$ is a transitive Milnor attractor
and supports an ergodic measure $\mu$ but is not uniquely ergodic.
In particular, at least one of the Lyapunov exponents (with respect to $\mu$) is positive.

Some of dynamical systems having chaotic attractors with densely intertwined basins of attraction, which we call it \emph{riddled basin}.
Riddled basins were introduced in 1992 by \cite{JI}.
In this case, a basin is riddled with holes (in a measure theoretical sense) of another basin.
To be more precise,
the basin of attraction $\mathcal{B}(A)$ of an attractor $A$ is \emph{riddled } if its complement $\mathcal{B}(A)^c$ intersects every disk in a set of positive measure.

This concept is generalized to the ``locally riddled basin". It considers the case where the basin of $A$ is open but local
normal unstable manifolds exist in a dense set in $A$. Precisely, a Milnor attractor $A$ has a \emph{locally riddled basin} if there exists a
neighbourhood $U$ of $A$ such that, for all $x \in A$ and $\varepsilon > 0$
\begin{equation}\label{l-riddel}
  m(B_{\varepsilon}(x)  \cap (\bigcap_{n\geq 0}F^{-n}(U))^c) >0.
\end{equation}

If there is another Milnor attractor $A^{\prime}$ such that $\mathcal{B}(A)^{c}$ in the definition of locally riddled basin may be replaced with $\mathcal{B}(A^{\prime})$, then we say that the basin of $A$ is riddled with the basin of $A^{\prime}$.
If $\mathcal{B}(A)$ and $\mathcal{B}(A^{\prime})$ are riddled with each other, we say that they are \emph{intermingled}.

The corresponding concepts can be defined for
repelling sets. An invariant transitive set $A$ is a
\emph{chaotic saddle} if there exists a neighborhood $U$ of $A$ such that $\mathcal{B}(A)\cap U \neq \emptyset$ but $m(\mathcal{B}(A)) = 0$.

Consider the case that the invariant set $A$ is contained in an invariant $n$-dimensional submanifold
$N$ of $\mathbb{R}^m$, with $n < m$. We say that $A$ is a \emph{normally repelling chaotic saddle} if $\mathcal{B}(A) \neq A$
and $\mathcal{B}(A) \subset N$. Note that in this case, however, $A$ is an attractor in the invariant subspace, but all
points not lying on this subspace eventually leave a neighborhood of $A$.
\subsection{Lyapunov exponents}
Assume $F$ is a smooth map defined on a smooth manifold $M$ and let $N \subset M$ be an $n$-dimensional embedded submanifold and forward invariant by $F$, with
$n < m$. Hence, for $x \in N$, one has that $d_xF(T_xN) \subset T_xF(N)$.
 We consider the restriction of $F$ to $N$, denoted by $F_{|N}$.
 Moreover, we assume that
$A$ is a chaotic attractor for $F$. We denote by $\mathcal{M}_F(A)$ and $\mathcal{E}_F(A)$ the sets of invariant probability
measures and ergodic measures supported in $A$, respectively. It is known that both $\mathcal{M}_F(A)$ and $\mathcal{E}_F(A)$
are non-empty (see \cite{W}).

For a vector $v \neq 0$ with base point $x$, \emph{the Lyapunov exponent} $\lambda(x,v)$ at the point $x$ in the direction
of $v$ is defined to be
\begin{equation}\label{11}
  \lambda(x,v)=\lim_{n \to \infty}\frac{1}{n}\log \|d_{x} F^n(v)\|_{T_{F^n(x)}M}
\end{equation}
whenever the limit exists.
Here, we have two kind of Lyapunov exponents for the chaotic invariant set $A$: the \emph{parallel Lyapunov exponents} which indicate the exponential rate of stretching on $A$ when $F$ is restricted to $N$
and the \emph{normal Lyapunov exponents} that present the exponential rate of expansion on $A$ in the normal direction denoted by $\lambda_{\parallel}$ and $\lambda_{\perp}$, respectively.
 Precisely, we define these Lyapunov exponents as follows.
Since $N$ is an embedded submanifold, we can take a smooth splitting of the tangent bundle $T M $ in a neighbourhood
of $N$ of the form $T_x M=T_x N\oplus (T_x N)^{\bot}$, when $x \in N$.
To simplify the notation, we write $TM_n := T_{F^n(x)}M$.
\begin{definition}\label{L1}
Given $x \in A$; $v \in T_{x}M=T_{x}N \oplus T_{x}N^{\perp}$, we define the \emph{parallel Lyapunov exponent} at
$x$ in the direction of $v$ to be
\begin{eqnarray}\label{22}
\lambda_{\parallel}(x,v) =\lim_{n\rightarrow \infty} \dfrac{1}{n} \ln\parallel \pi_{(T N_{n})} \circ d_{x}F^{(n)} \circ \pi_{T N_{0}} (v)  \parallel_{TM },
\end{eqnarray}
where $\pi_V$ is the orthogonal projection onto a subspace $V$.
Similarly, we define the \emph{normal Lyapunov exponent} at $x$ in the direction of $v$ to be
\begin{eqnarray}\label{33}
\lambda_{\perp}(x,v) =\lim_{n\rightarrow \infty}  \dfrac{1}{n} \ln\parallel \pi_{(T N_{n})^{\perp}} \circ d_{x}F^{(n) }\circ \pi_{(T N_{0})^{\perp}} (v)  \parallel_{TM_{n} }.
\end{eqnarray}
\end{definition}
Here, we are interested in Sinai-Ruelle-Bowen (or SRB) measures which are a special type of invariant measures, see  \cite{MA}.

Let $A$ be an asymptotically stable attractor under $F_{|N}$.
An invariant ergodic probability measure $\mu$ is called an $SRB$ measure for $A$ if its support is
$A$ and has absolutely continuous conditional measures on unstable manifolds
(with respect to the Riemannian measure).

An attractor $A$ is an \emph{SRB-attractor} if it supports an SRB measure.
Since $A$ is an asymptotically stable attractor under $F_{|N}$, hence, it is the closure of the union of unstable
manifolds on $N$. Note that the existence of an SRB measure supported on $A$
implies the absolute continuity of the stable foliation of $A$, see \cite{Pugh}.
By a result from \cite{Pugh}, we have the following:

Assume $\mu$ is an SRB measure for $F_{|N}$. Given a neighborhood $U$ of $A$ there is a set $B(\mu) \subset U$, we call it the \emph{basin} of $\mu$,
with positive Lebesgue measure such that for all $x \in B(\mu)$
and all continuous functions $\phi :N \to \mathbb{R}$ one has that
\begin{equation*}
  \lim_{n \to \infty}\frac{1}{n}\sum_{i=0}^{n-1} \phi(F^i|_N(x))=\int_A \phi d\mu.
\end{equation*}
\begin{remark}
Take an open set $D$ with $\mu_{SRB}(D)=\mu_{SRB}(\text{Cl}(D))$, then
\begin{equation*}
 \lim_{n \to \infty}\frac{1}{n}\sum_{i=0}^{n-1} \chi_D(F^i|_N(x))= \mu_{SRB}(D),
\end{equation*}
 for a.e. $x$ in the basin $B(\mu)$.
 \end{remark}
Given an ergodic invariant probability measure $\mu \in \mathcal{E}_{F_{|N}}(A)$, the normal Lyapunov
exponents $\lambda_{\bot}^1(\mu)< \cdots <\lambda_{\bot}^s(\mu)$
exists and are constants in a set $B_\mu$ of full $\mu$-measure.
We define
\begin{equation}\label{mm}
  \lambda_{min}:=\inf \bigcup_{\mu \in \mathcal{E}_{F_{|N}}(A)}\{\lambda_{\bot}^i(\mu)\}, \ \Lambda_{max}:=\sup \bigcup_{\mu \in \mathcal{E}_{F_{|N}}(A)}\{\lambda_{\bot}^i(\mu)\}.
\end{equation}
Let $\mu$ be an $F$-invariant ergodic probability measure supported in $A$, with
normal Lyapunov exponents $\lambda_{\bot}^1(\mu)< \cdots <\lambda_{\bot}^s(\mu)$.
The \emph{normal stability index} $\Lambda_\mu$ of $\mu$ is
\begin{equation}\label{in}
 \Lambda_\mu:=\lambda_{\bot}^s(\mu).
\end{equation}
 We recall the next result from \cite{JI}.
\begin{theorem}\label{thm11}
Assume $A$ is an SRB attractor for $F_{|N}$ with $\Lambda_{SRB} < 0$, where $\Lambda_{SRB}$ is defined by (\ref{in}) for SRB measure $\mu_{SRB}$. Then $m(\mathcal{B}(A)) > 0$.
Furthermore, $A$ is an essential attractor provided that $A$ is either uniformly hyperbolic or $\mu_{SRB}$ is absolutely continuous with respect
to Riemannian measure on $N$.
\end{theorem}
\subsection{Blowout bifurcations}
Let $M$ be a smooth Riemannian manifold and $N \subset M$ be an
$n$-dimensional embedded submanifold and forward invariant by $F$, with $n < m$.
Assume that there is a parameter $\beta$ that varies the transverse dynamics without changing the dynamics on the invariant subspace $N$.
We call the parameter $\beta$ a \emph{normal parameter} \cite{PJ}.
Moreover, we assume that the normal Lyapunov exponents
vary continuously with $\beta$ and that $A$ is an asymptotically
stable attractor for the restriction map $F_{|N}$.

Let $A$ be a chaotic attractor that supports an SRB measure $\mu_{SRB}$. Then, the sign of $\Lambda_{SRB}$ determines whether $A$ attracts
or repels infinitesimal perturbations in the direction transverse to $N$.
If $\Lambda_{SRB}<0$, $A$ attracts trajectories transversely in
the phase space and hence, $A$ is also an
attractor of the whole phase space. When $\Lambda_{SRB}>0$,
trajectories in the vicinity of $A$ are repelled away from it. That is, $A$ is transversely unstable
and is not an attractor of the whole phase space.
Thus a bifurcation occurs when $\Lambda_{SRB}$ crosses zero, the so-called \emph{blowout bifurcation}.

Blowout bifurcations are classified as either hysteretic (subcritical) or non-hysteretic (supercritical).
In a hysteretic blowout, riddled basins before the blowout give rise to a hard loss of stability \cite{AP, AS}, after blowout, almost all points near the invariant subspace eventually move
away, never to return. However, a non-hysteretic blowout giving a soft loss of stability to an on-off intermittent attractor \cite{AP, NE}.

It was shown that \cite{PJ} the occurrence of ``locally riddled basins" is exhibited near the blowout bifurcation.
\subsection{Step skew products}
Let $\Sigma_2^+=\{1,2\}^{\mathbb{N}}$ and equip $\Sigma_2^+$ with the topology generated by the base of cylinder sets
\begin{equation*}
  [\alpha_0, \cdots, \alpha_{n-1}]=\{\omega \in \Sigma_2^+: \omega_j=\alpha_j, \ \text{for} \ j=0, \cdots, n-1\}.
\end{equation*}
Let $\sigma: \Sigma_2^+ \to \Sigma_2^+$ be the (left) shift operator defined by $(\sigma \omega)_i=\omega_{i+1}$ for $\omega = (\omega_i)_0^\infty$.

Here, we take the Bernoulli measure $\mathbb{P}^+$ on $\Sigma_2^+$ where the symbols $1, 2$ have probability $p_1= p_2=1/2$.
Indeed, the Bernoulli measure $\mathbb{P}^+$ on $\Sigma_2^+$ is determined by its values on cylinders;
\begin{equation}\label{Ber}
 \mathbb{P}^+([\alpha_0, \cdots, \alpha_{n-1}])=\prod_{j=0}^{n-1}p_{\alpha_j}.
\end{equation}

Let $G^+$ be a step skew product map with the fiber maps $g_1$ and $g_2$ defined by
\begin{equation}\label{s0}
 G^+ :\Sigma_2^+ \times \mathbb{I} \to \Sigma_2^+ \times \mathbb{I}, \quad (\omega,x)\mapsto (\sigma \omega, g_{\omega_0}(x)).
\end{equation}

Denote by $\mathcal{B}$ the Borel sigma-algebra on $\mathbb{I}$. For a skew product $G^+$ with fiber maps $g_1$ and $g_2$, a measure $m$ on $\mathbb{I}$ and any
$\mathcal{B}$-measurable set $A$, we denote the push-forward measure of $m$ by $g_i m$, $i=1,2$, in which
\begin{equation*}
  g_i m(A)=m(g_i^{-1}(A)).
\end{equation*}
We recall that a probability measure $m$ on the interval $\mathbb{I}$ for $G^+$ is \emph{stationary} if
it satisfies
\begin{equation*}
  m=\sum_{i=1}^2 p_i g_i m.
\end{equation*}

The natural extension of $G^+$ is obtained when the shift acts on two sided time $\mathbb{Z}$; this yields a
skew product system $G : \Sigma_2 \times \mathbb{I} \to \Sigma_2 \times \mathbb{I}$ with $\Sigma_2=\{1,2\}^{\mathbb{Z}}$ and given by the same
expression
\begin{equation}\label{s1}
 G(\omega,y)=(\sigma \omega, g_{\omega_0}(y)).
\end{equation}
Write $\pi:\Sigma_2 \to \Sigma_2^+$ for the natural projection $(\omega_n)_{-\infty}^\infty \mapsto (\omega_n)_{0}^\infty$.
The Borel sigma-algebra $\mathcal{F}^+$ on $\Sigma_2^+$ yields a sigma-algebra $\mathcal{F} = \pi^{-1}\mathcal{F}^+$ on $\Sigma_2$.
Write $\mathbb{P}$ for the Bernoulli measure on $\Sigma_2$, defined analogues to $\mathbb{P}^+$ (see (\ref{Ber})).

The invariant measure $\mu^+=\mathbb{P}^+ \times m$ for $G^+$ gives rise to an invariant measure $\mu$ for $G$, with marginal $\mathbb{P}$.
Invariant measures for $G^+$ with marginal $\mathbb{P}^+$ and invariant measures for $G$ with
marginal $\mathbb{P}$ are in one to one relationship, see \cite{Ar}.

A measure $\mu$ on $\Sigma_2 \times \mathbb{I}$ with marginal $\mathbb{P}$ has conditional measures $\mu_{\omega}$ on the fibers $\{\omega\}\times \mathbb{I}$
such that
\begin{equation}\label{co1}
  \mu(A)=\int_{\Sigma_2}\mu_{\omega}A_{\omega}d \mathbb{P}(\omega)
\end{equation}
for each measurable set $A$, where $A_{\omega}=A \cap (\{\omega\}\times \mathbb{I})$.
For $\mathbb{P}$-almost all $\omega$, the conditional measures are given as the weak star limit
\begin{equation}\label{co2}
  \mu_\omega=\lim_{n \to \infty}g_{\sigma^{-n}\omega}^n m
\end{equation}
of push-forwards of the stationary measure.


\section{\textbf{The occurrence of locally riddled basin}}
In this section, we concentrate on studying the two parameters family $F_{\beta_1,\beta_2}: \mathbb{I} \times \mathbb{I} \to \mathbb{I} \times \mathbb{I}$ of skew product maps as defined in (\ref{ss}).
By varying the parameters $\beta_1$ and $\beta_2$, we will investigate the occurrence of locally riddled basins for the family $F_{\beta_1,\beta_2}$.

We recall a normal parameter of the system—one that preserves the dynamics on the
invariant submanifold but varies it in the rest of the phase space, first introduced in \cite{PJ}.
We observe that the parameters $\beta_i$ vary the transverse
dynamics without changing the dynamics on the invariant
subspaces $N_0$ or $N_1$, so, they are normal parameters.

In our model, the expanding map $f$ exhibits a chaotic attractor which supports an absolutely continuous invariant ergodic measure (a.c.i.m) whose density is bounded and bounded away from zero (see \cite{RL}).
In particular, this measure is absolutely continuous with respect to Lebesgue.
By this fact and invariance of the subspaces $N_i$, the restriction of $F_{\beta_1,\beta_2}$ to these invariant subspaces possess chaotic attractors $A_i$, $i=0,1$, with the basin of attraction $\mathcal{B}(A_i)$.
In particular, these attractors are $SRB$ attractors.

Note that, in our context, the normal dynamics is continuously dependent on normal parameters $\beta_i$.
Also, the invariant subspaces have codimension 1 in the phase space $\mathbb{I} \times \mathbb{I}$ and there is only one normal direction.
Hence, we can discuss the transitions between the invariant sets $A_i$ being
 a locally riddled basin attractor and a chaotic saddle.

By Theorem \ref{thm11}, since $\mu_{SRB}$ is absolutely continuous with respect
to an Riemannian measure on the subspace $N_i$, $i=1,2$, the subset $A_i$ is an essential attractor if the normal Lyapunov exponent is negative, see also \cite{JI}.

In general, for the occurrence of a riddled basin, we require that the parallel Lyapunov
exponent $\lambda_{\parallel}$ is positive but the maximal normal Lyapunov exponent $\Lambda_{SRB}$ is slightly negative.

First, we describe the dynamics of the particular skew product $F_{\beta_1,\beta_2}$ whose fiber maps $g_{i,\beta_1,\beta_2}$, $i=1,2$, given by (\ref{fiber1}).
For almost every point in $A_i$ the parallel Lyapunov exponent for the dynamics
on $N_i$ is
\begin{align}
L_{\Vert}=1/2 \ln 2+1/2 \ln 2=\ln 2.
\end{align}
In particular, it is positive and does not depend on $\beta_i$, $i=1,2$.
By a straightforward computations, the normal Lyapunov exponents at typical points in $A_i$ are given by
\begin{align}\label{n0}
L_{\perp, \beta_1,\beta_2}(0)=1/2 \ln(g_{1,\beta_1,\beta_2}^{\prime}(0))+1/2 \ln(g_{2,\beta_1,\beta_2}^{\prime}(0))=1/2 \ln (1-\beta_1)+ 1/2 \ln (1+\beta_2)
\end{align}
and
\begin{align}\label{n1}
L_{\perp, \beta_1,\beta_2}(1)=1/2 \ln(g_{1,\beta_1,\beta_2}^{\prime}(1))+1/2 \ln(g_{2, \beta_1,\beta_2}^{\prime}(1))=1/2 \ln \beta_1 + 1/2 \ln (2-\beta_2).
\end{align}
We observe that, by (\ref{n0}) and (\ref{n1}), these normal Lyapunov exponents vary continuously with $\beta_i$. In Figure \ref{fig:8} the plot of $L_{\perp, \beta_1,\beta_2}(i)$, $i=0,1$,
are illustrated by varying $\beta_1$ and for fixed $\beta_2=1/2$. It demonstrates the continuous dependence of normal Lyapunov exponents with respect to $\beta_1$ in this case.

 \begin{figure}[h!]
\begin{center}
    $\begin{array}{c}
  \includegraphics[scale=0.4]{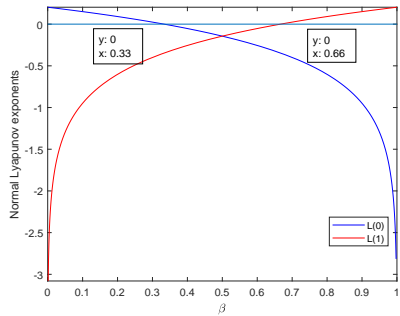}
    \end{array}$
\end{center}
  \caption{\small Normal Lyapunov exponents vary continuously with $\beta_i$. The blue curve depicts $L_{\perp, \beta_1,\beta_2}(0)$ by varying $\beta_1=\beta$ and for fixed value $\beta_2=1/2$,
  while the red curve depicts $L_{\perp, \beta_1,\beta_2}(1)$.}
 \label{fig:8}
   \end{figure}

 Given an ergodic measure $\mu$, let $G_\mu$ be the set of generic points of $\mu$. That is
 \begin{equation*}
   G_\mu=\{x \in A:\frac{1}{n}\sum_{i=0}^{n-1}\delta_{f^j(x)}\to \mu\}
 \end{equation*}
 where convergence is in the weak$^*$ topology. For any $\alpha > 0$ define
 \begin{equation}\label{alpha}
   G_\alpha:=\bigcup_{\mu \in \mathcal{E}_f(A),\Lambda_\mu \geq \alpha}G_\mu.
 \end{equation}

\begin{proposition}\cite[Proposition~3.19]{PJ}\label{p1}
Suppose $F : M \to M$ is a $C^{1+\alpha}$ map leaving the embedded submanifold $N$
invariant, and that $A$ is an asymptotically stable chaotic attractor for $F|_N$.
Let $\Lambda_{max}$, $\lambda_{min}$ and $\Lambda_{SRB}$ be given by (\ref{mm}) and (\ref{in}). Then, under $F : M \to M$
\begin{enumerate}
  \item [$(1)$] If $\Lambda_{SRB} < 0 < \Lambda_{max}$ then $A$ is a Milnor (essential) attractor. If in addition there
exists $\alpha > 0$ with $G_\alpha$ dense in $A$, then $A$ has a locally riddled basin.
  \item [$(2)$] If $ \lambda_{min}< 0 < \Lambda_{SRB}$, $\mu_{SRB}$-almost all Lyapunov exponents are non-zero and $m(\bigcup_{\mu \neq \mu_{SRB}}G_\mu)=0$, where $m$ is the Riemannian volume on $N$,
 then $A$ is a chaotic saddle.
\end{enumerate}
\end{proposition}
Note that, by \cite[Remark~3.4]{ABS}, if $codim(N)=1$, as in our setting, there is only one normal direction. In this case $\lambda_\mu=\Lambda_\mu$ for
all ergodic $\mu$, and the normal spectrum depends smoothly on normal parameters.

Define $\rho : \mathbb{I} \to \{-1,1\}$ by
\begin{align}\label{code0}
\rho(x) =\left\{\begin{array}{cc}
-1 & $ for$ \quad0\leq x\leq 1/2,\\
1 & $ for$ \quad 1/2< x\leq1.
\end{array}\right.
\end{align}
Let $x$ be a periodic point of $f$ of period $n$ and let $\mathcal{O}(x)=\{x_i: f(x_{i-1})=x_i, \ i=1, \dots, n, \quad  \text{with} \quad x_0=x\}$.
Take
\begin{equation}\label{per}
 Per^+(f):=\{x \in Per(f): \sum_{i=0}^{n-1}\rho(x_i) > 0\},
\end{equation}
 where $Per(f)$ denotes the set of all periodic orbits of $f$.
By \cite[Proposition~3.1]{Go}, there is a semi-conjugacy between the shift map $\sigma : \Sigma_2^+ \to \Sigma_2^+$ and the
doubling map $f$. By this fact, the following result is immediate.
\begin{lemma}\label{lem1}
The subset $Per^+(f)$ is dense in each chaotic attractor $A_i$, $i=0,1$, within the invariant subspaces $N_i$.
\end{lemma}
\begin{theorem}\label{th1}
Let $F_{\beta_1,\beta_2}\in \mathcal{F}$ be a skew product of the form (\ref{ss}) whose fiber maps $g_{i,\beta_1,\beta_2}$, $i=1,2$, given by (\ref{fiber1}) and $\beta_1, \beta_2 \in (0,1)$.
Then the following holds:
\begin{enumerate}
  \item [$(a)$] if $\beta_1 >1/2$ or ($\beta_1 \leq 1/2$ $\&$ $ \beta_2 <\frac{1}{1-\beta_1}-1$), then $A_0$ is a Milnor (essential) attractor and has a locally riddled basin;
  \item [$(b)$] if $\beta_1 < 1/2$ or ($ \beta_1\geq 1/2$ $\&$ $\beta_2 >2-\frac{1}{\beta_1}$), then $A_1$ is a Milnor (essential) attractor and has a locally riddled basin;
  \item [$(c)$] if $\beta_1<1/2$ and $\beta_2 > \frac{1}{1-\beta_1}-1$ then $A_0$ is a chaotic saddle;
  \item [$(d)$] if $\beta_1 >1/2$, $\beta_2<2- \frac{1}{\beta_1}$ then $A_1$ is a chaotic saddle.
\end{enumerate}
\end{theorem}
\begin{proof}

Let $N_i$, $i=0,1$, be given by (\ref{subspaces}) and consider the restriction $F_{\beta_1,\beta_2}|_{N_i}$.
By definition, since there is only one normal direction, the normal stability index $\Lambda_\mu(A_i)$ of an ergodic invariant probability measure
 $\mu \in \mathcal{E}_{F_{\beta_1,\beta_2}|_{N_i}}(A_i)$ is equal to
the normal Lyapunov exponent $\lambda_{\mu}(A_i)$.

Note that, for $0 \leq x \leq 1/2$, we have
\begin{equation*}
 d_{(x,0)}F_{\beta_1, \beta_2}=\begin{pmatrix}
2 & \quad 0\\
0& \quad dg_{1,\beta_1}(0)\\
\end{pmatrix}
=\begin{pmatrix}
2 & \quad 0\\
0& \quad 1-\beta_1\\
\end{pmatrix}
\end{equation*}
and for $1/2 < x \leq 1$, we have
\begin{equation*}
 d_{(x,0)}F_{\beta_1, \beta_2}=\begin{pmatrix}
2 & \quad 0\\
0& \quad dg_{2,\beta_2}(0)\\
\end{pmatrix}
=\begin{pmatrix}
2 & \quad 0\\
0& \quad 1+\beta_2\\
\end{pmatrix}
\end{equation*}
where $g_{i,\beta_i}$, $i=1,2$, given by (\ref{fiber1}).

Hence, the normal stability index is computed as follows:
\begin{eqnarray}\label{Lambd}
  \lambda_\mu(A_0)=\Lambda_{\mu}(A_0) = \int_{A_0 \cap [0,1/2]}\log (1-\beta_1) d\mu(x ) + \int_{A_0 \cap (1/2,1]}\log (1+\beta_2) d\mu(x ).
 \end{eqnarray}
Therefore,
\begin{eqnarray}\label{Lambd1}
  \lambda_\mu(A_0)=\Lambda_{\mu}(A_0) = \mu(A_0 \cap [0,1/2])\log (1-\beta_1) + \mu(A_0 \cap (1/2,1])\log (1+\beta_2).
 \end{eqnarray}

Note that, by (\ref{Lambd1}), for each invariant measure $\mu$, $\Lambda_{\mu}$ is finite.
Additionally, $\Lambda_{\mu}$ is smoothly dependent on normal parameters $\beta_1$ and $\beta_2$.

The base map $f$ given by (\ref{base}) is a piecewise expanding map.
By definition of $F_{\beta_1,\beta_2}$ and since $N_0$ is one dimensional, we conclude that
$F_{\beta_1,\beta_2}|_{A_0}$ is also piecewise expanding. This fact implies that $F_{\beta_1,\beta_2}|_{A_0}$ has a
Lebesgue-equivalent ergodic invariant measure (see \cite{ABS, W}); this corresponds to the desired $\mu_{SRB}(A_0)$ (see Subsection 4.3 of \cite{ABS}).
By this fact and (\ref{Lambd1}), $$\Lambda_{SRB}(A_0)=1/2 \log (1-\beta_1) + 1/2 \log (1+\beta_2).$$
Note that $\Lambda_{SRB}(A_0)=L_{\perp, \beta_1,\beta_2}(0)$, where $L_{\perp, \beta_1,\beta_2}(0)$ is the normal Lyapunov exponent given by (\ref{n0}).
Also, it characterizes \cite{JI} evolution transverse to the $x$-axis. If it
is negative, the invariant set $A_0$ is a Milnor attractor.
Simple computations show that, for $\beta_1 >1/2$ or ($\beta_1\leq 1/2$ $\&$ $ \beta_2 <\frac{1}{1-\beta_1}-1$), $\Lambda_{SRB}(A_0)<0$ and hence $A_0$ is a Milnor (essential) attractor.
Take the invariant dirac measure $\mu_1$ supported on the fixed point at $(1,0)$. Using (\ref{Lambd1}), $\Lambda_{\mu_1}(A_0)=\log(1+\beta_2)$ which is positive.
By this fact, $0< \Lambda_{\mu_1}(A_0) \leq \Lambda_{max}(A_0)$.
These computations show that $\Lambda_{SRB}(A_0)<0 < \Lambda_{max}(A_0)$.

By Lemma \ref{lem1}, the set $Per(f)^+$ is dense in $A_0$. By definition of $Per(f)^+$ and (\ref{Lambd1}), the dirac measure supported on $\mathcal{O}(x)$, for $x \in Per(f)^+$, has positive normal stability index.
By this fact and taking suitable dirac measures supported on $\mathcal{O}(x)$, for $x \in Per(f)^+$, we may find $\alpha>0$ such that
 $G_\alpha$ given by (\ref{alpha}) is dense in $A_0$.
By these observations and statement (1) of Proposition \ref{p1}, $A_0$ has a locally riddled basin and the proof of $(a)$ is finished.

Let $\mu_2$ be the invariant dirac measure supported on the fixed point at $(0,0)$. Using (\ref{Lambd1}), $\Lambda_{\mu_2}(A_0)=\log(1-\beta_1)$ which is negative. By this fact, $\lambda_{min}(A_0)\leq \Lambda_{\mu_2}(A_0) < 0$.
Also, it is easy to see that for $\beta_1<1/2$ and $\beta_2 > \frac{1}{1-\beta_1}-1$, $\Lambda_{SRB}(A_0)>0$.
By these facts, for $\beta_1<1/2$ and $\beta_2 > \frac{1}{1-\beta_1}-1$, we get $\lambda_{min}(A_0)<0<\Lambda_{SRB}(A_0)$. Since $\mu_{SRB}(A_0)$ is equivalent to Lebesgue and whose support is $A_0$,
 $m(\bigcup_{\mu \neq \mu_{SRB}(A_0)}G_\mu)=0$, where $m$ is the Lebesgue measure on $N_0$. Clearly, $\mu_{SRB}(A_0)$-almost all Lyapunov exponents are non-zero. By these facts  and statement (2) of Proposition \ref{p1}, $A_0$ is a chaotic saddle which verifies statement $(c)$.

We apply similar arguments to prove $(b)$ and $(d)$. Indeed,
for $0 \leq x \leq 1/2$, we have
\begin{equation*}
 d_{(x,1)}F_{\beta_1, \beta_2}=\begin{pmatrix}
2 & \quad 0\\
0& \quad dg_{1,\beta_1}(1)\\
\end{pmatrix}
=\begin{pmatrix}
2 & \quad 0\\
0& \quad \beta_1\\
\end{pmatrix}
\end{equation*}
and for $1/2 < x \leq 1$, we have
\begin{equation*}
 d_{(x,1)}F_{\beta_1, \beta_2}=\begin{pmatrix}
2 & \quad 0\\
0& \quad dg_{2,\beta_2}(1)\\
\end{pmatrix}
=\begin{pmatrix}
2 & \quad 0\\
0& \quad 2- \beta_2\\
\end{pmatrix}
\end{equation*}
where $g_{i,\beta_i}$, $i=1,2$, given by (\ref{fiber1}).
Hence, the normal stability index is computed as follows:
\begin{eqnarray}\label{Lambd}
  \lambda_\mu(A_1)=\Lambda_{\mu}(A_1) = \int_{A_1 \cap [0,1/2]}\log (\beta_1) d\mu(x ) + \int_{A_1 \cap (1/2,1]}\log (2- \beta_2) d\mu(x ).
 \end{eqnarray}
As a consequence,
\begin{eqnarray}\label{Lambd1}
  \lambda_\mu(A_1)=\Lambda_{\mu}(A_1) = \mu(A_1 \cap [0,1/2])\log (\beta_1) + \mu(A_1 \cap (1/2,1])\log (2- \beta_2).
 \end{eqnarray}
As above, $F_{\beta_1,\beta_2}|_{A_1}$ has a
Lebesgue-equivalent ergodic invariant measure; this corresponds to the desired $\mu_{SRB}(A_1)$.
By this fact and (\ref{Lambd1}), for the attractor $A_1$, $\Lambda_{SRB}(A_1)=1/2 \log (\beta_1) + 1/2 \log (2- \beta_2)$.
Note that $\Lambda_{SRB}(A_1)=L_{\perp, \beta_1,\beta_2}(1)$, where $L_{\perp, \beta_1,\beta_2}(1)$ is given by (\ref{n1}).
Simple computations show that, for $\beta_1 < 1/2$ or ($ \beta_1\geq1/2$ $\&$ $\beta_2 >2-\frac{1}{\beta_1}$), $\Lambda_{SRB}(A_1)<0$ and hence $A_1$ is a Milnor (essential) attractor.
Take the invariant dirac measure $\nu_1$ supported on the fixed point at $(1,1)$. Using (\ref{Lambd1}), $\Lambda_{\nu_1}(A_1)=\log(2-\beta_2)$ which is positive.
By this fact, $0< \Lambda_{\nu_1} (A_1)\leq \Lambda_{max}(A_1)$.
These computations show that $\Lambda_{SRB}(A_1)<0 < \Lambda_{max}(A_1)$.
By argument applied in $(a)$, we may find $\alpha>0$ such that
 $G_\alpha$ given by (\ref{alpha}) is dense in $A_1$.
By these observations and statement (1) of Proposition \ref{p1}, $A_1$ has a locally riddled basin and the proof of $(b)$ is finished.

Let $\nu_2$ be the invariant dirac measure supported on the fixed point at $(0,1)$. Using (\ref{Lambd1}), $\Lambda_{\nu_2}(A_1)=\log(\beta_1)$ which is negative.
 By this fact, $\lambda_{min}(A_1)\leq \Lambda_{\nu_2}(A_1) < 0$.
Also, it is easy to see that for $\beta_1 >1/2$, $\beta_2<2- \frac{1}{\beta_1}$, $\Lambda_{SRB}(A_1)>0$.
By these facts, if $\beta_1 >1/2$, $\beta_2<2- \frac{1}{\beta_1}$, then $\lambda_{min}(A_1)<0<\Lambda_{SRB}(A_1)$. Since $\mu_{SRB}(A_1)$ is equivalent to Lebesgue and whose support is $A_1$,
 $m(\bigcup_{\mu \neq \mu_{SRB}(A_1)}G_\mu)=0$. By these facts  and statement (2) of Proposition \ref{p1}, $A_1$ is a chaotic saddle which verifies statement $(d)$.
\end{proof}

\begin{figure}[h!]
\begin{center}
   $\begin{array}{cc}
  \includegraphics[scale=0.4]{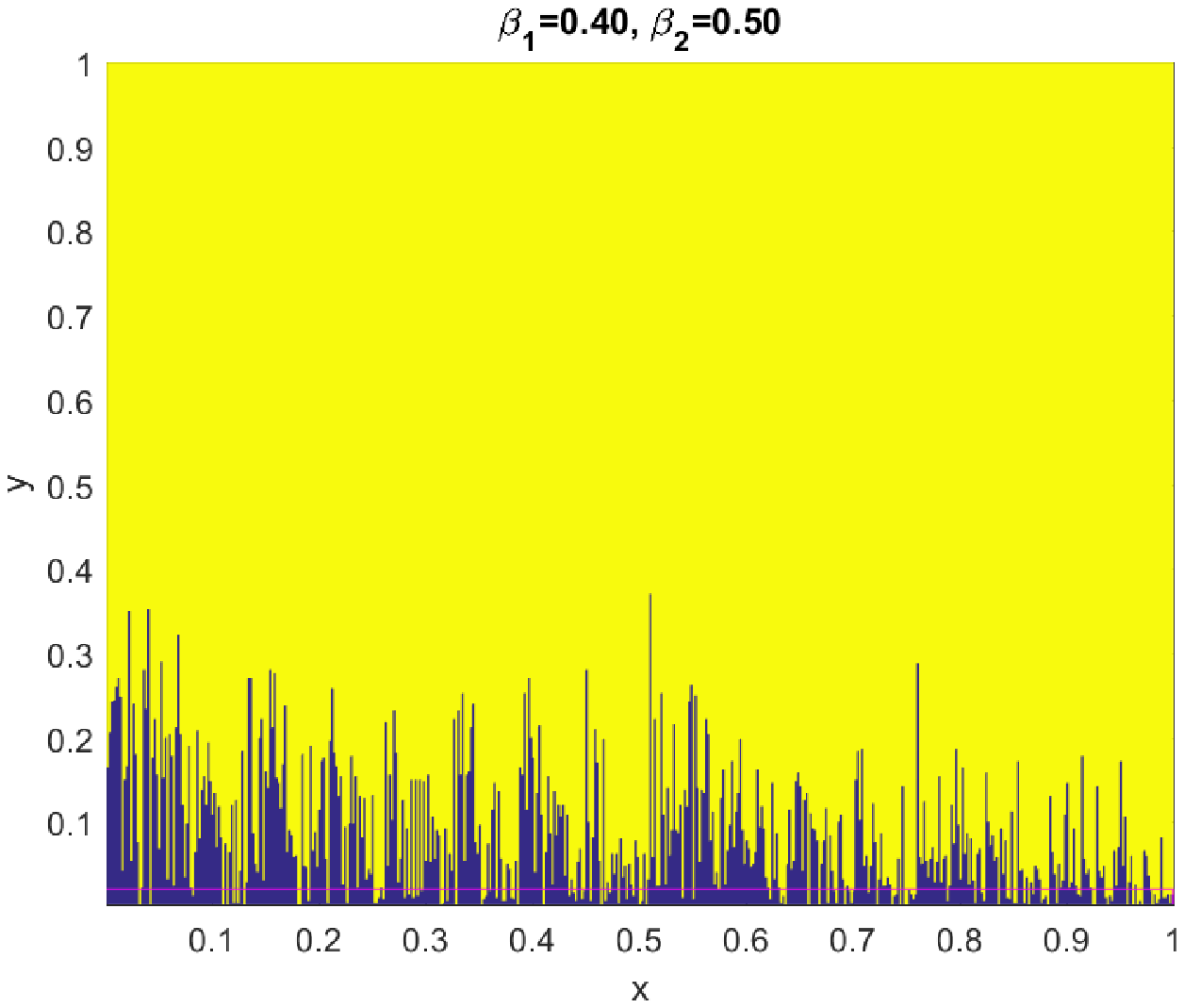}
  \includegraphics[scale=0.4]{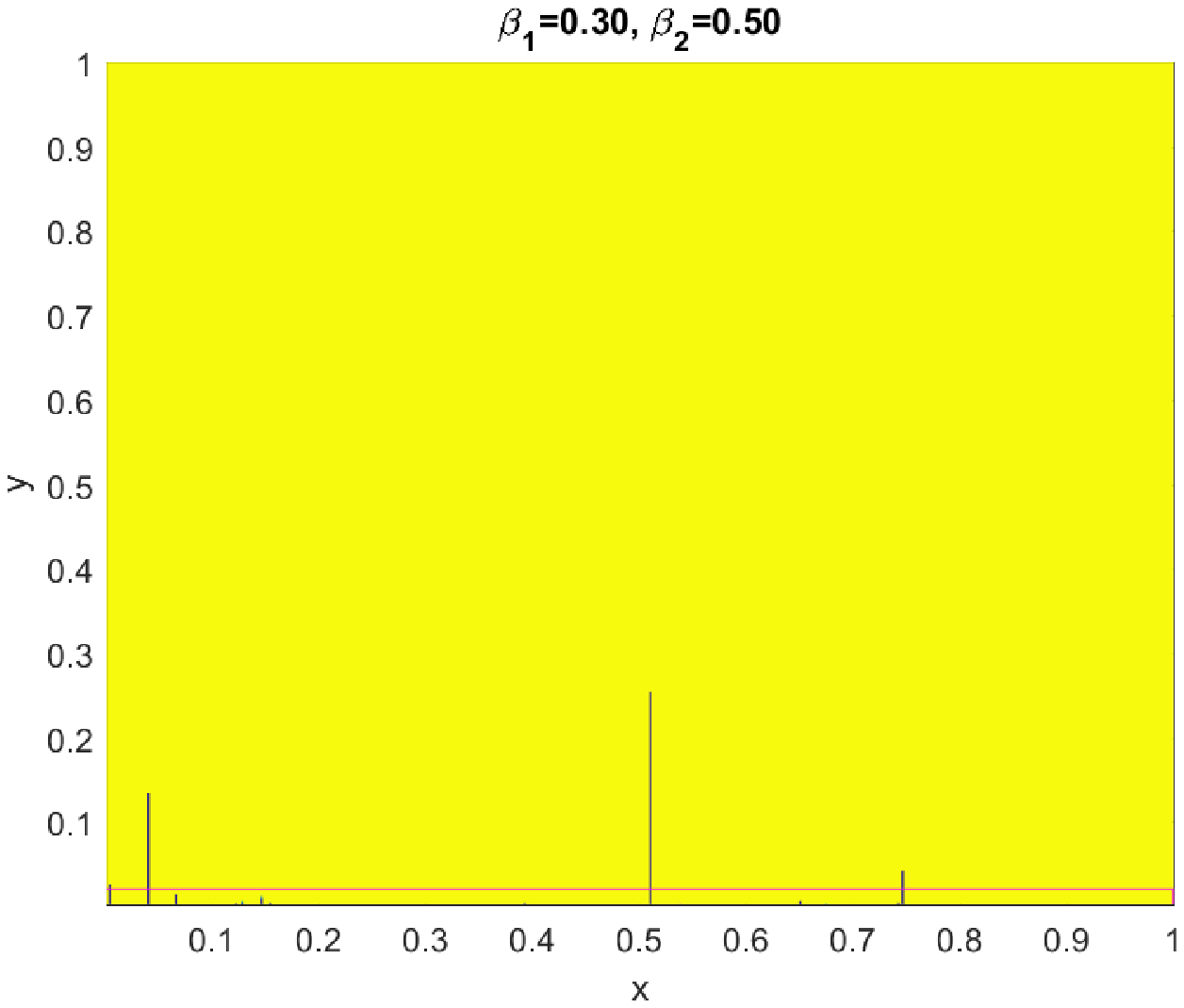}
  \end{array}$
\end{center}
  \caption{\small The basins of attraction for the attractors $A_0$ and $A_1$. The blue region in each figure corresponds to the basin of attraction $\mathcal{B}(A_0)$ while the yellow region corresponds to
the basin of attraction $\mathcal{B}(A_1)$.
  The left frame depicts the basins of attraction of $A_0$ and $A_1$, for $\beta_1=0.4$, $\beta_2=0.5$. In this case $A_i$, $i=0,1$,
  are Milnor attractors and $A_0$ has locally riddled basin. The right frame shows the basins of attraction
of $A_0$ and $A_1$ for $\beta_1=0.3$, $\beta_2=0.5$. In this case $A_0$ is a chaotic saddle and the basin $\mathcal{B}(A_0)$ has zero measure, while $A_1$ is a Milnor attractor.}
 \label{fig:01}
   \end{figure}

\begin{figure}[h!]
\begin{center}
    $\begin{array}{cc}
  \includegraphics[scale=0.4]{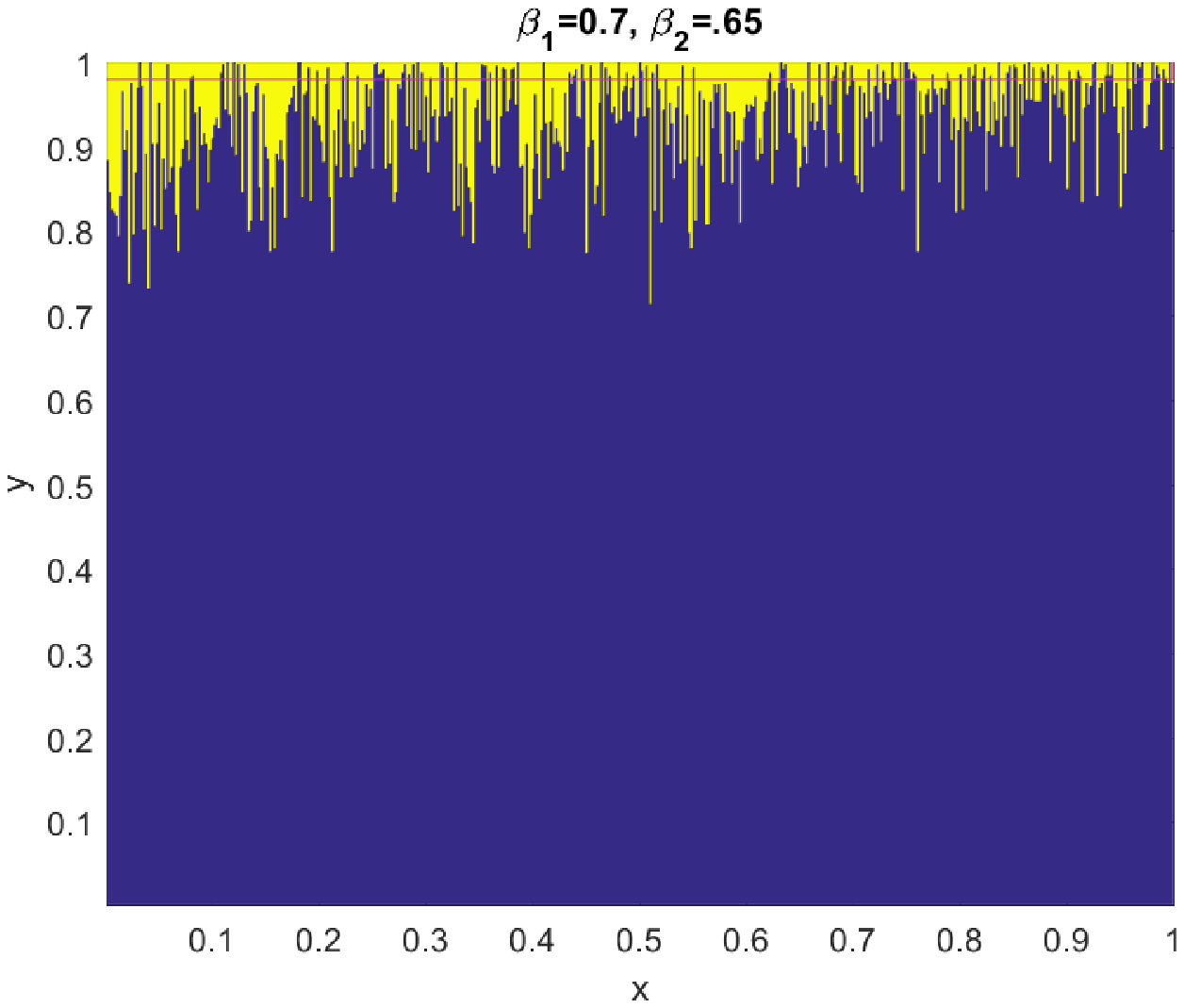}
  \includegraphics[scale=0.4]{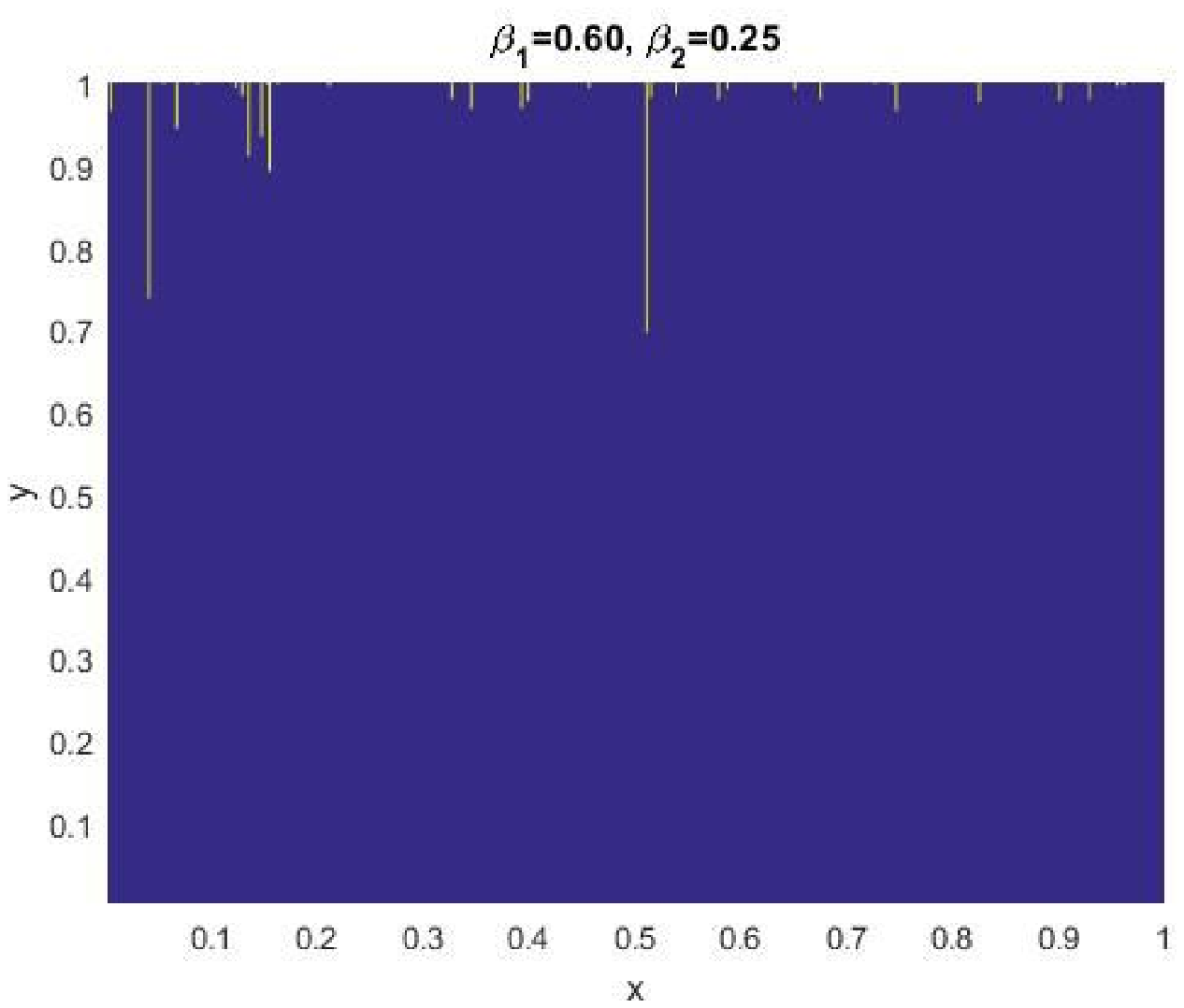}
  \end{array}$
\end{center}
  \caption{\small The basins of attraction for the attractors $A_0$ and $A_1$. The blue region in each figure corresponds to the basin of attraction $\mathcal{B}(A_0)$ while the yellow region corresponds to
the basin of attraction $\mathcal{B}(A_1)$.
  The left frame depicts the basins of attraction of $A_0$ and $A_1$, for $\beta_1=0.7$, $\beta_2=0.65$. In this case $A_i$, $i=0,1$,
  are Milnor attractors and $A_1$ has locally riddled basin. The right frame shows the basins of attraction
of $A_0$ and $A_1$ for $\beta_1=0.6$, $\beta_2=0.25$. In this case $A_1$ is a chaotic saddle and the basin $\mathcal{B}(A_1)$ has zero measure, while $A_0$ is a Milnor attractor.}
 \label{fig:02}
   \end{figure}

Let
\begin{equation}\label{region1}
 \Gamma_S:=\{(\beta_1,\beta_2) : \beta_1, \beta_2 \in (0,1), \ 2- \frac{1}{\beta_1} < \beta_2 < \frac{1}{1-\beta_1}-1 \}.
\end{equation}
If we set
\begin{align}\label{ff}
  \Gamma_{\frac{1}{2}}=\{(\beta_1,\beta_2) : \beta_2 = \frac{1}{2}, \quad \beta_1 \in (\frac{1}{3},\frac{2}{3})\}
\end{align}
then $\Gamma_{\frac{1}{2}}\subset  \Gamma_S$. Therefore, $ \Gamma_S$ is a nonempty open region.
\begin{corollary}\label{cor1}
For each $(\beta_1,\beta_2)\in  \Gamma_S$, the both invariant sets $A_0$ and $A_1$ are Milnor attractors and have locally riddled basins.
\end{corollary}

\begin{figure}[h!]
\begin{center}
    $\begin{array}{c}
  \includegraphics[scale=0.5]{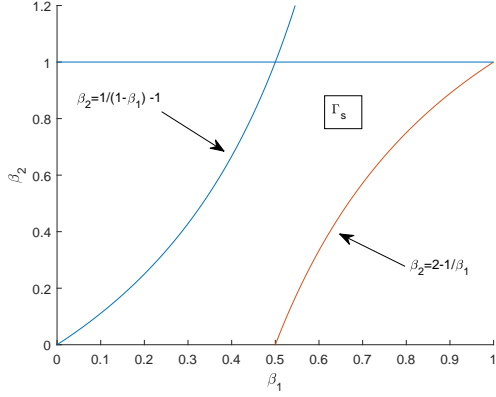}
    \end{array}$
\end{center}
  \caption{\small The open region $\Gamma_S$, where for each $(\beta_1,\beta_2)\in  \Gamma_S$, the both invariant sets $A_0$ and $A_1$ are Milnor attractors.}
 \label{fig:88}
   \end{figure}

We extend the above results to the general case. Let $F_{\beta_1,\beta_2}\in \mathcal{F}$ be a skew product of the form (\ref{ss}) whose fiber maps $g_{i,\beta_1,\beta_2}$, $i=1,2$, satisfy conditions $(I1)$-$(I3)$.
For $j=0,1$, we take
\begin{equation}\label{alpha0}
 \alpha_{\beta_1 \beta_2}^j:=dg_{1,\beta_1}(j)dg_{2,\beta_2}(j),
\end{equation}
where $dg_{i,\beta_i}(j)$, $i=1,2$, are the derivatives of the fiber maps $g_{i,\beta_i}$ at the point $j$.
\begin{theorem}\label{thm2}
Let $F_{\beta_1,\beta_2}\in \mathcal{F}$ be a skew product of the form (\ref{ss}) whose fiber maps $g_{i,\beta_1,\beta_2}$, $i=1,2$, satisfy conditions $(I1)$-$(I3)$.
Then the following statements hold:
\begin{enumerate}
  \item [$(a)$] If there exists $\beta_1^r \in (0,1)$ such that for each $\beta_1 < \beta_1^r$, the sign of $\log(\alpha_{\beta_1 \beta_2}^0)$ changes from negative to positive by varying $\beta_2$,
   then, there exists a smooth function $\beta_1 \mapsto \xi(\beta_1)$ with $\xi(\beta_1^r)=1$ such that
   \begin{enumerate}
     \item [$(i)$]  if $\beta_1 < \beta_1^r$ and $ \beta_2 <\xi(\beta_1)$, then $A_0$ is a Milnor (essential) attractor and has a locally riddled basin,
     \item [$(ii)$] if $\beta_1 < \beta_1^c$ and $ \beta_2 > \xi(\beta_1)$, then $A_0$ is a chaotic saddle;
   \end{enumerate}
     \item [$(b)$] If there exists $\beta_1^\ell \in (0,1)$ such that for each $\beta_1 > \beta_1^\ell$, the sign of $\log(\alpha_{\beta_1 \beta_2}^1)$ changes from positive to negative by varying $\beta_2$,
   then, there exists a smooth function $\beta_1 \mapsto \zeta(\beta_1)$ with $\zeta(\beta_1^\ell)=0$ such that
   \begin{enumerate}
     \item [$(i)$]  if $\beta_1 > \beta_1^\ell$ and $ \beta_2 >\zeta(\beta_1)$, then $A_1$ is a Milnor (essential) attractor and has a locally riddled basin;
     \item [$(ii)$] if $\beta_1 > \beta_1^\ell$ and $ \beta_2 < \zeta(\beta_1)$, then $A_1$ is a chaotic saddle.
   \end{enumerate}
 \end{enumerate}
\end{theorem}
\begin{proof}
To prove, we closely follow the proof of Theorem \ref{th1} and omit some details.
As we have seen before there is only one normal direction, so the normal stability index $\Lambda_{\mu^i}$, $i=0,1$, of an ergodic invariant probability measure
 $\mu^i \in \mathcal{E}_{F_{\beta_1,\beta_2}|_{N_i}}(A_i)$ is equal to
the normal Lyapunov exponent $\lambda_{{\mu}^i}$.

For $0 \leq x \leq 1/2$,
\begin{equation*}
 d_{(x,y)}F_{\beta_1, \beta_2}=\begin{pmatrix}
2 & \quad 0\\
0& \quad dg_{1,\beta_1}(y)\\
\end{pmatrix}
\end{equation*}
and for $1/2 < x \leq 1$,
\begin{equation*}
 d_{(x,y)}F_{\beta_1, \beta_2}=\begin{pmatrix}
2 & \quad 0\\
0& \quad dg_{2,\beta_2}(y)\\
\end{pmatrix}
\end{equation*}
where $g_{i,\beta_i}$, $i=1,2$, are the fiber maps of $F_{\beta_1, \beta_2}$.

Hence, the normal stability index $\Lambda_{\mu^i}$ for the attractor $A_i$, $i=0,1$, is computed as follows:
\begin{eqnarray}\label{Lambd0}
  \lambda_{\mu^i}=\Lambda_{{\mu}^i} = \int_{A_i \cap [0,1/2]}\log (dg_{1,\beta_1}(i)) d\mu^i(x ) + \int_{A_i \cap (1/2,1]}\log (dg_{2,\beta_2}(i)) d\mu^i(x ).
 \end{eqnarray}
Therefore,
\begin{eqnarray}\label{Lambd2}
  \lambda_{\mu^i}=\Lambda_{{\mu}^i} = \mu^i(A_i \cap [0,1/2])\log (dg_{1,\beta_1}(i)) + \mu^i(A_i \cap (1/2,1])\log (dg_{2,\beta_2}(i)).
 \end{eqnarray}
Note that $F_{\beta_1,\beta_2}|_{A_i}$, $i=0,1$, is piecewise expanding, hence, it has a
Lebesgue-equivalent ergodic invariant measure which is the desired $\mu_{SRB}^i$.
By this fact and (\ref{Lambd2}), $$\Lambda_{SRB}^i=1/2 \log (dg_{1,\beta_1}(i)) + 1/2 \log (dg_{2,\beta_2}(i)).$$

By (\ref{Lambd2}), for each invariant measure $\mu^i$, $\Lambda_{{\mu}^i}$ is smoothly dependent on the normal parameters $\beta_1$ and $\beta_2$.
This fact and the hypothesis of statements $(a)$ and $(b)$ imply that there exist smooth functions $\beta_1 \mapsto \xi(\beta_1)$ and $\beta_1 \mapsto \zeta(\beta_1)$ such that they satisfy the following properties:
\begin{equation}\label{1}
  \Lambda_{SRB}^o <0 \quad \text{if} \quad  \beta_1 < \beta_1^r \ \& \  \beta_2 <\xi(\beta_1), \quad \text{and} \quad  \Lambda_{SRB}^o >0 \quad \text{if} \quad \beta_1 < \beta_1^r \ \& \  \beta_2 >\xi(\beta_1),
\end{equation}
\begin{equation}\label{3}
\Lambda_{SRB}^1 <1 \quad \text{if} \quad \beta_1 > \beta_1^\ell \ \& \ \beta_2 >\zeta(\beta_1) , \quad \text{and} \quad \Lambda_{SRB}^1 >0 \quad \text{if} \quad \beta_1 > \beta_1^\ell \ \& \ \beta_2 <\zeta(\beta_1).
\end{equation}

Take the invariant dirac measure $\nu_1^0$ supported on the fixed point at $(1,0)$. Using (\ref{Lambd2}), $\Lambda_{\nu_1^0}=\log(dg_{2,\beta_2}(0))$ which is positive by condition $(I3)$.
 By this fact, $0< \Lambda_{\nu_1^0} \leq \Lambda_{max}^0$.
These computations show that $\Lambda_{SRB}^0 <0< \Lambda_{max}^0$ for the attractor $A_0$.
Using (\ref{Lambd2}) and Lemma \ref{lem1}, we apply the argument used in the proof of Theorem \ref{th1} to find $\alpha^0>0$ such that the subset $G_{\alpha^0}$ given by (\ref{alpha}) is dense in $A_0$.
By these observations and statement (1) of Proposition \ref{p1}, $A_0$ has a locally riddled basin and the proof of the first statement of $(a)$ is finished.

Similarly, we take the invariant dirac measure $\nu_0^1$ supported on the fixed point at $(0,1)$. Using (\ref{Lambd2}), $\Lambda_{\nu_0^1}=\log(dg_{1,\beta_1}(1))$
which is positive by condition $(I2)$. By this fact, $0< \Lambda_{\nu_0^1} \leq \Lambda_{max}^1$.
These computations show that $\Lambda_{SRB}^1<0 < \Lambda_{max}^1$.
As above, using Lemma \ref{lem1} and (\ref{Lambd2}),  we can apply the argument used in the proof of Theorem \ref{th1} to find $\alpha^1>0$ such that the subset $G_{\alpha^1}$ given by (\ref{alpha}) is dense in $A_1$.
By these observations and statement (1) of Proposition \ref{p1}, $A_1$ has a locally riddled basin and the proof of the first statement of $(b)$ is finished.

Let $\nu_0^0$ be the invariant dirac measure supported on the fixed point at $(0,0)$. Using (\ref{Lambd2}), $\Lambda_{\nu_0^0}=\log(dg_{1,\beta_1}(0))$ which is negative by condition $(I2)$. By this fact,
 $\lambda_{min}^0\leq \Lambda_{\nu_0^0} < 0$.
Also, by (\ref{1}) for $ \beta_2 > \xi(\beta_1)$, $\Lambda_{SRB}^0>0$.
By these facts, for $ \beta_2 > \xi(\beta_1)$, we get $\lambda_{min}^0<0<\Lambda_{SRB}^0$. Since $\mu_{SRB}^0$ is equivalent to Lebesgue and whose support is $A_0$,
 $m(\bigcup_{\mu \neq \mu_{SRB}^0}G_\mu)=0$, where $m$ is the Lebesgue measure on $N_0$. Clearly, $\mu_{SRB}^0$-almost all Lyapunov exponents are non-zero. By these facts  and statement (2) of Proposition \ref{p1}, $A_0$ is a chaotic saddle which verifies the second statement of $(a)$.

Let $\nu_0^1$ be the invariant dirac measure supported on the fixed point at $(0,1)$. Using (\ref{Lambd1}), $\Lambda_{\nu_0^1}=\log(dg_{1,\beta_1}(0))$ which is negative by condition $(I2)$. By this fact,
$\lambda_{min}^1\leq \Lambda_{\mu_0^1} < 0$.
Also, by (\ref{3}), for $\beta_2< \zeta(\beta_1)$, $\Lambda_{SRB}^1>0$.
By these facts, if $\beta_2> \zeta(\beta_1)$, then $\lambda_{min}^1<0<\Lambda_{SRB}^1$. Since $\mu_{SRB}^1$ is equivalent to Lebesgue and whose support is $A_1$,
 $m(\bigcup_{\mu \neq \mu_{SRB}^1}G_\mu)=0$. By these facts  and statement (2) of Proposition \ref{p1}, $A_1$ is a chaotic saddle which verifies the second statement of $(b)$.
\end{proof}
Let
\begin{equation}\label{region2}
 \Gamma_G:=\{(\beta_1,\beta_2) : \beta_1, \beta_2 \in (0,1), \quad \zeta(\beta_2)<\beta_2 <\xi(\beta_1)\}.
\end{equation}
Then, $\Gamma_S \subset \Gamma_G$, and hence, $\Gamma_G$ is a nonempty open region in $\beta_1 \beta_2$-plane.
\begin{corollary}\label{cor2}
For each $(\beta_1,\beta_2)\in  \Gamma_G$, the both invariant sets $A_0$ and $A_1$ are Milnor attractors and have locally riddled basins.
\end{corollary}
\section{\textbf{Random walk model}}
In this section, we introduce a random walk model which is topologically semi-conjugate to
the skew product system $F_{\beta_1,\beta_2}$ given by (\ref{ss}). This allows us to define a fractal boundary between the initial conditions leading to each of the two attractors $A_1$ and $A_2$.
\begin{definition}
Let $\mathcal{S}$ be the set of step skew product systems of the form
\begin{equation}\label{s}
G_{\beta_1,\beta_2}^+:\Sigma_2^+ \times \mathbb{I} \to \Sigma_2^+ \times \mathbb{I}, \ \  G_{\beta_1,\beta_2}^+(\omega,y)=(\sigma \omega, g_{\omega_0,\beta_{\omega_0}}(y)),
\end{equation}
where $\mathbb{I}=[0,1]$, $\beta_1,\beta_2 \in (0,1)$ such that the fiber maps $g_{1,\beta_1}$ and $g_{2,\beta_2}$ are strictly increasing $C^2$ diffeomorphisms given by (\ref{fiber})
satisfying conditions $(I1) - (I3)$.
\end{definition}
We will pick the diffeomorphisms $g_{1,\beta_1}$ and $g_{2,\beta_2}$ randomly, independently at each iterate, with
positive probabilities $p_1$ and $p_2 = 1-p_1$. This corresponds to taking a Bernoulli
measure on $\Sigma_2^+$ from which we pick $\omega$.
The obtained random compositions
\begin{equation}\label{com}
g_{\omega, \beta_1,\beta_2}^n(y)=g_{\omega_{n-1},\beta_{\omega_{n-1}}} \circ \dots \circ g_{\omega_0,\beta_{\omega_0}}(y), \quad \text{for} \ n\geq 1, \quad g_{\omega, \beta_1,\beta_2}^0(y)=y
\end{equation}
form a random walk on the interval.

We define the \emph{standard measure} $s$ on $\Sigma_2^+ \times \mathbb{I}$ by the product of Bernoulli
measure $\mathbb{P}^+$ and the Lebesgue measure on $\mathbb{I}$.

Given a skew product $G_{\beta_1,\beta_2}^+ \in \mathcal{S}$, the \emph{normal Lyapunov exponent} of $G_{\beta_1,\beta_2}^+$ at a point $(\omega, y) \in \Sigma_2^+ \times \mathbb{I}$ is
\begin{equation}\label{L}
 \lim_{n \to \infty}\frac{1}{n} \ln (g^{\prime}_{\omega_{n-1},\beta_{\omega_{n-1}}}(g_{\omega,\beta_1,\beta_2}^{n-1}(y))\dots g_{\omega_0,\beta_{\omega_0}}^{\prime}(y))
 =\lim_{n \to \infty}\frac{1}{n}\sum_{i=0}^{n-1}\ln (g_{\omega_i,\beta_{\omega_i}}^{\prime}(g_{\omega,\beta_1,\beta_2}^i(y))),
\end{equation}
in case the limit exists. Since $y = 0, 1$ are fixed points of $g_{i,\beta_i}$, $i = 1, 2$, by Birkhoff's
ergodic theorem, we obtain for $y = 0, 1$ that
\begin{equation}\label{FL}
 L_{\beta_1,\beta_2}(y)=\lim_{n \to \infty}\frac{1}{n}\sum_{i=0}^{n-1}\ln (g_{\omega_i,\beta_{\omega_i}}^{\prime}(g_{\omega,\beta_1,\beta_2}^i(y)))=\int_{\Sigma_2^+}\ln(g_{\omega_0,\beta_{\omega_0}}^{\prime}(y))d\nu^+(\omega)=\sum_{i=1}^2p_i \ln g_{i,\beta_i}^{\prime}(y)
\end{equation}
for $\mathbb{P}^+$-almost all $\omega \in \Sigma_2^+$.

The subspaces
\begin{equation*}
 \mathbb{A}_0:=\Sigma_2^+ \times \{0\}, \ \ \text{and} \ \ \mathbb{A}_1:=\Sigma_2^+ \times \{1\}
\end{equation*}
are invariants by $G^+_{\beta_1,\beta_2}$, for each $\beta_1,\beta_2 \in (0,1)$.
The basins are
\begin{equation*}
  \mathcal{B}_{\beta_1,\beta_2}(\mathbb{A}_0) := \{(\omega, y) : d((G_{\beta_1,\beta_2}^+)^n(\omega, y),\mathbb{A}_0)\to 0 \ \text{as} \ n \to \infty \},
\end{equation*}
\begin{equation*}
  \mathcal{B}_{\beta_1,\beta_2}(\mathbb{A}_1) := \{(\omega, y) : d((G_{\beta_1,\beta_2}^+)^n(\omega, y),\mathbb{A}_1) \to 0 \ \text{as} \ n \to \infty \},
\end{equation*}
where $d(y,A) = \inf_{z \in A} \|y-z\|$, for each subset $A \subset \mathbb{I}$.

First, we consider a specific example of a step skew product system from $\mathcal{S}$ with the fiber maps
\begin{equation}\label{diff}
g_{1,\beta_1}(y)=y+y(1-y)(y-\beta_1), \ \ g_{2,1/2}(y)=y-y(1-y)(y-1/2),
\end{equation}
with $\beta_2=\frac{1}{2}$ and illustrate time series of random walks for some values of $\beta_1$.
Using (\ref{FL}), we get
\begin{equation}\label{LB}
L_{\beta_1,1/2}(0)=1/2 \ln (1-\beta_1)+1/2 /ln 3/2, \  \text{and} \ L_{\beta_1,1/2}(1)=1/2 \ln \beta_1+1/2 /ln 3/2.
\end{equation}
Note that $L_{\beta_1, 1/2}(0)<0$ and $L_{\beta_1, 1/2}(1)<0$ for each $1/3 < \beta_1 < 2/3$.

\begin{figure}[h!]
\begin{center}
    $\begin{array}{cc}
  \includegraphics[scale=0.4]{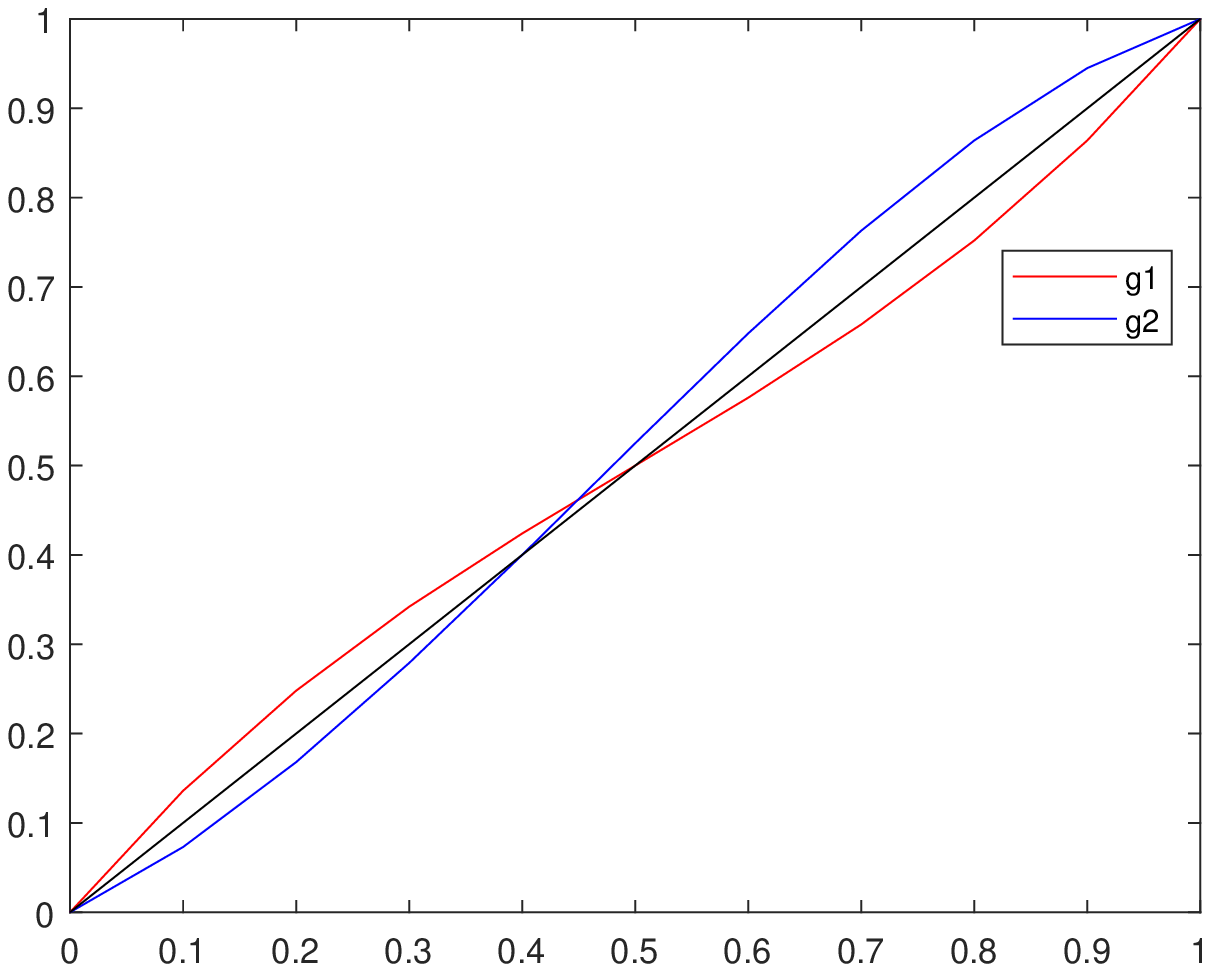}
  \includegraphics[scale=0.4]{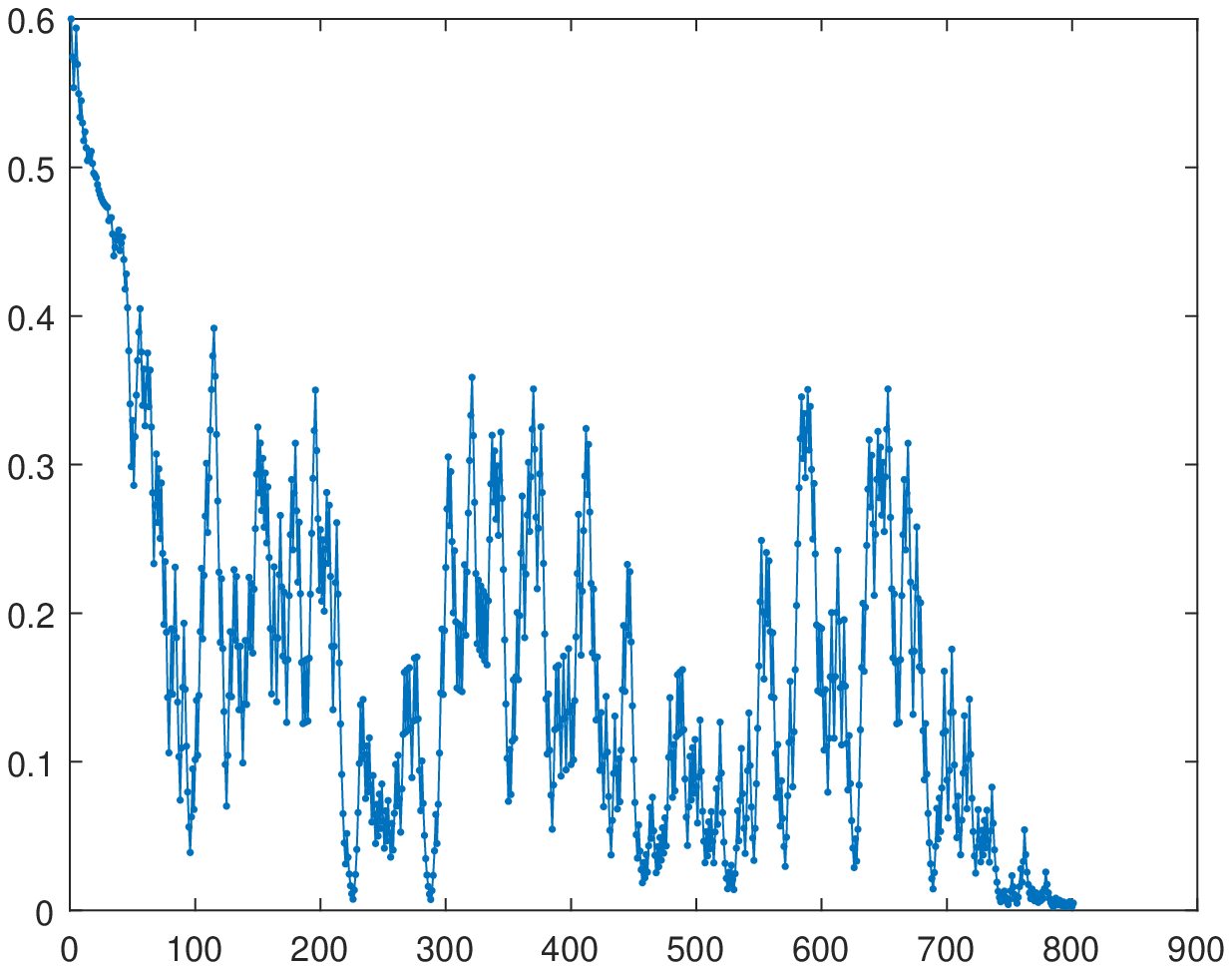}
  \end{array}$
\end{center}
  \caption{\small The left frame depicts the graphs of $g_{1,\beta_1}$, $g_{2,1/2}$, for $\beta_1=2/5$.
The right frame shows a time series of the random walk generated by
$g_{1,\beta_1}$, $g_{2,1/2}$, for $\beta_1=2/5$, both picked with probability $1/2$.}
 \label{fig:1}
   \end{figure}

 \begin{figure}[h!]
\begin{center}
    $\begin{array}{cc}
  \includegraphics[scale=0.4]{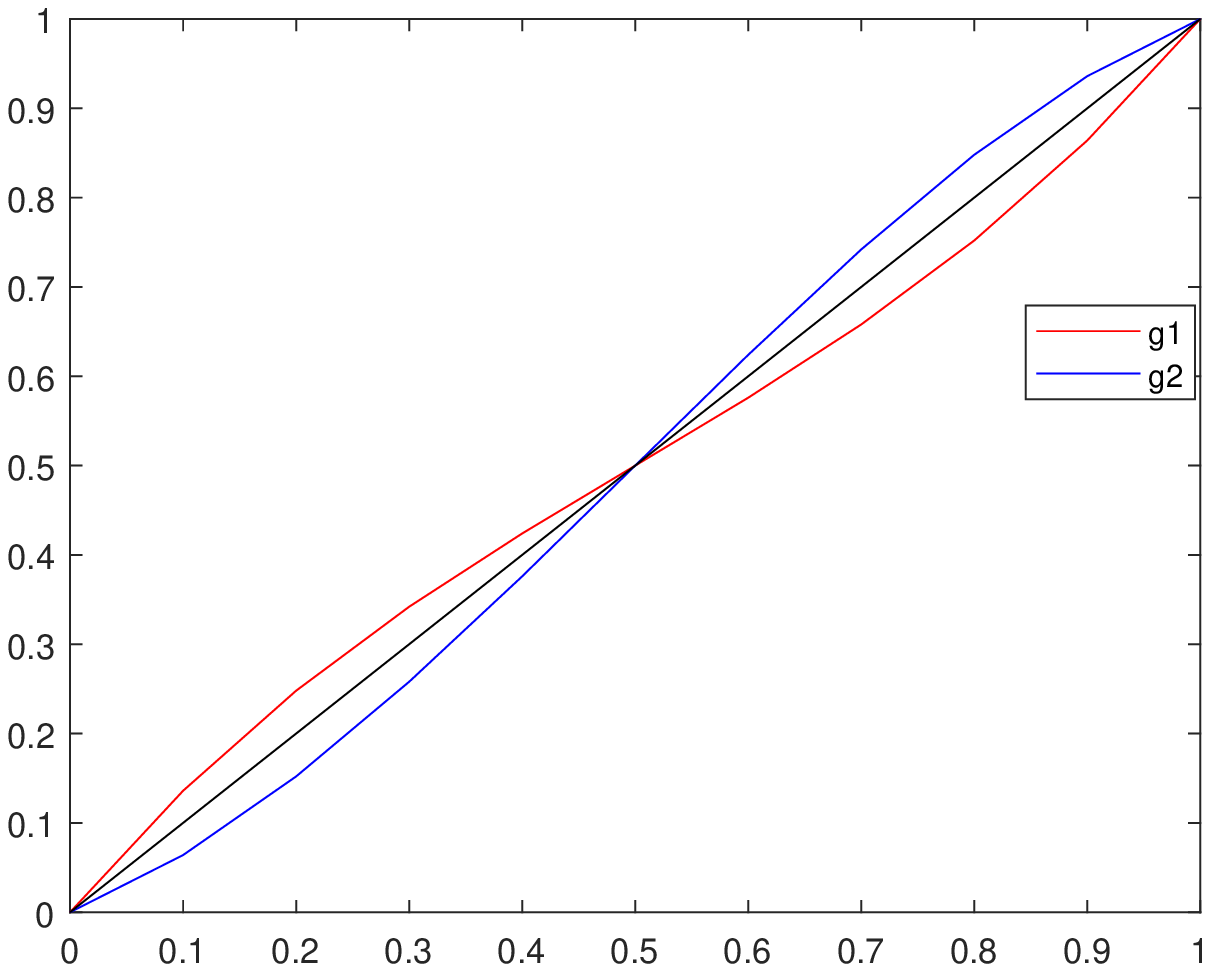}
  \includegraphics[scale=0.4]{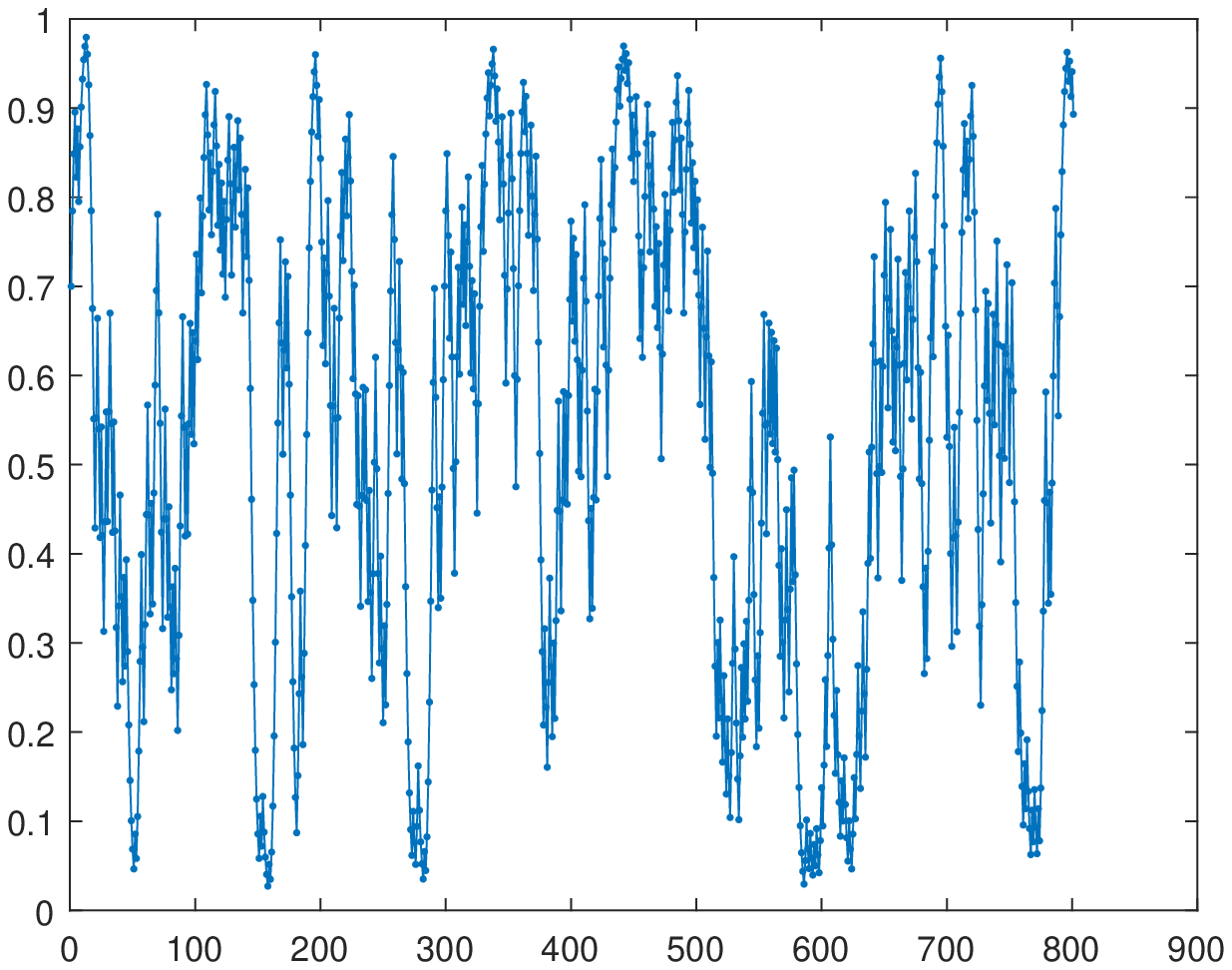}
  \end{array}$
\end{center}
  \caption{\small The left frame depicts the graphs of $g_{1,\beta_1}$, $g_{2,1/2}$, for $\beta_1=1/2$.
The right frame shows a time series of the random walk generated by
$g_{1,\beta_1}$, $g_{2,1/2}$, for $\beta_1=1/2$, both picked with probability $1/2$.}
 \label{fig:2}
   \end{figure}

 \begin{figure}[h!]
\begin{center}
    $\begin{array}{cc}
  \includegraphics[scale=0.4]{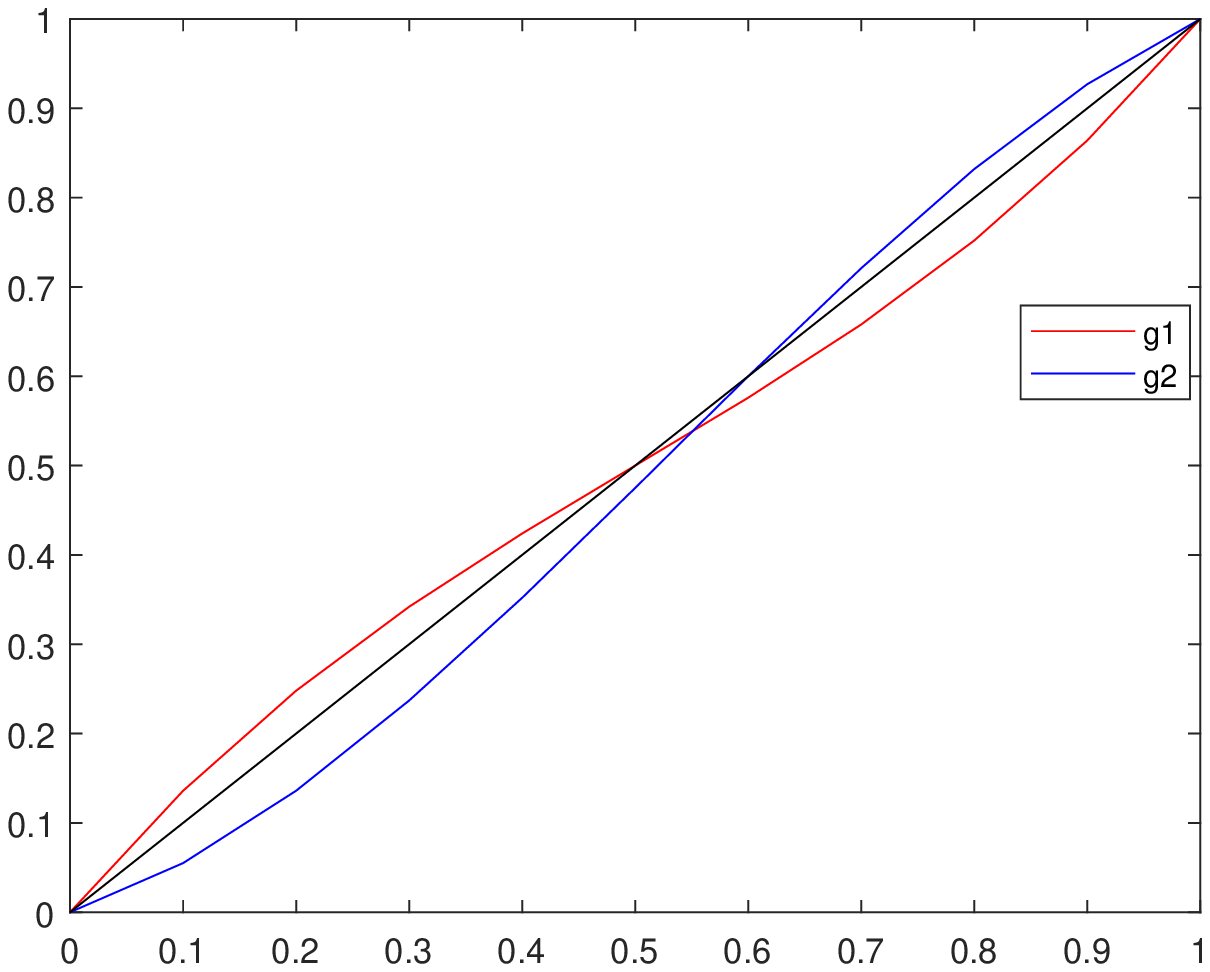}
  \includegraphics[scale=0.4]{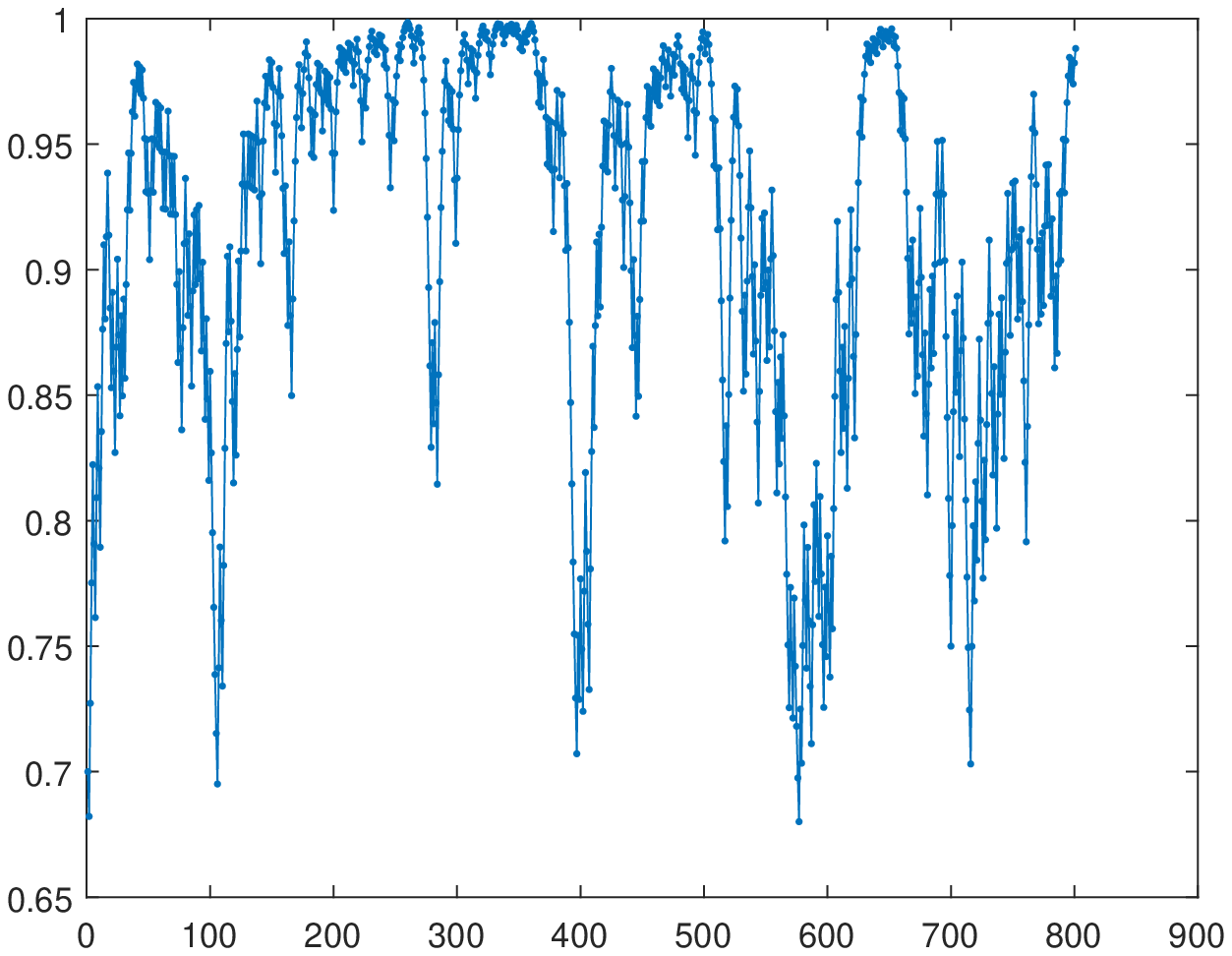}
  \end{array}$
\end{center}
  \caption{\small The left frame depicts the graphs of $g_{1,\beta_1}$, $g_{2,1/2}$, for $\beta_1=3/5$.
The right frame shows a time series of the random walk generated by
$g_{1,\beta_1}$, $g_{2,1/2}$, for $\beta_1=3/5$, both picked with probability $1/2$.}
 \label{fig:3}
   \end{figure}

The graphs of $g_{1,\beta_1}$, $g_{2,1/2}$ are illustrated in the left frames of Fig. \ref{fig:1}, Fig. \ref{fig:2} and Fig. \ref{fig:3}, for parameter values $\beta_1=2/5, 1/2, 3/5$, respectively.
The right panels of these figures show time series of the random walk generated by
$g_{1,\beta_1}$, $g_{2,1/2}$, for $\beta_1=2/5,1/2,2/3$, both picked with probability $1/2$.
For these values of $\beta_1$, the normal Lyapunov exponents at $0$ and $1$ are negative.
\begin{theorem}\label{t1}
Let $G^+_{\beta_1,\beta_2} \in \mathcal{S}$ be a step skew product whose fiber maps defined by (\ref{fiber1}) and $(\beta_1,\beta_2)\in \Gamma_S$, where $\Gamma_S$ is given by (\ref{region1}).
 Then, the sets $\mathbb{A}_0$ and $\mathbb{A}_1$ are Milnor attractors for $G^+_{\beta_1,\beta_2}$.
They attract sets of positive standard measure and the union of the basins $\mathcal{B}_{\beta_1,\beta_2}(\mathbb{A}_i)$, $i=0,1$, has full standard measure.
\end{theorem}
\begin{proof}
By (\ref{FL}) and the definition of $\Gamma_S$, $L_{\beta_1,\beta_2}(0)<0$ and $L_{\beta_1,\beta_2}(1)<0$.

\begin{lemma}
Let $(\beta_1, \beta_2)\in \Gamma_S$ and take
\begin{equation*}
  r_{\beta_1,\beta_2}(\omega)=\sup \{y \in \mathbb{I}: \lim_{n \to \infty} g_{\omega,\beta_1,\beta_2}^n(y)=0\}, \quad s_{\beta_1,\beta_2}(\omega)=\inf \{y \in \mathbb{I}: \lim_{n \to \infty} g_{\omega,\beta_1,\beta_2}^n(y)=1\},
\end{equation*}
where $g_{\omega,\beta_1,\beta_2}^n(y)$ is given by (\ref{com}). Then $r_{\beta_1,\beta_2}(\omega)>0$ and $s_{\beta_1,\beta_2}(\omega)<1$ for $\mathbb{P}^+$ almost all $\omega \in \Sigma_2^+$.
\end{lemma}
\begin{proof}
The lemma is a counterpart of \cite[Lemma 3.1]{GH} to our setting.
In the proof of \cite[Lemma 3.1]{GH}, the authors only uses the fact that the normal Lyapunov exponents at the point 0 and 1 are negative and applying Birkhoff's ergodic theorem.
Since $L_{\beta_1,\beta_2}(0)<0$ and $L_{\beta_1,\beta_2}(1)<0$, the lemma is proved.
\end{proof}

Since, by the previous lemma, for each $(\beta_1, \beta_2)\in \Gamma_S$, the function $r_{\beta_1,\beta_2}$ is positive almost everywhere, thus,
 the basin $\mathcal{B}_{\beta_1,\beta_2}(\mathbb{A}_0)$ has
positive standard measure. Since, the function $s_{\beta_1,\beta_2}$ is less than 1 almost everywhere, the same holds for the basin of $\mathcal{B}_{\beta_1,\beta_2}(\mathbb{A}_1)$.
By these facts, the basins $\mathcal{B}_{\beta_1,\beta_2}(\mathbb{A}_i)$ of $G_{\beta_1,\beta_2}^+$-invariant sets $\mathbb{A}_i$, $i=0,1$, have positive standard measures.
In particular, they are Milnor attractors. This proves the first statement of the theorem.

Write $\mathcal{P}$ for the space of probability measures on $\mathbb{I}$, equipped with the weak
star topology and define $\mathcal{T} : \mathcal{P} \to \mathcal{P}$ by
\begin{equation*}
  \mathcal{T}m=\sum_{i=1}^2 p_i g_{i,\beta_i} m.
\end{equation*}
We recall that a stationary measure is a fixed point of $\mathcal{T}$.
Note that a probability measure $m$ is a stationary measure if and only if $\mu^+=\mathbb{P}^+ \times m$
is an invariant measure of $G_{\beta_1,\beta_2}^+$ with marginal $\mathbb{P}^+$ on $G_{\beta_1,\beta_2}^+$, see \cite[Lemma A.2]{GH}.
We say that $m$ is ergodic if $\mathbb{P}^+ \times m$ is ergodic for $G_{\beta_1,\beta_2}^+$.

To proceed proving Theorem \ref{t1}, we prove
that there is a boundary between basins $\mathcal{B}_{\beta_1,\beta_2}(\mathbb{A}_0)$ and $\mathcal{B}_{\beta_1,\beta_2}(\mathbb{A}_1)$.
Indeed, we show that there is an invariant measurable graph $\psi_{\beta_1,\beta_2} : \Sigma_2 \to \mathbb{I}$ that separates the basins: for $\mathbb{P}$-almost all $\omega$,
\begin{align}\label{xi}
\lim_{n \to \infty}g_{\omega,\beta_1,\beta_2}^n(y) =\left\{\begin{array}{cc}
0 & $ if$ \quad y< \psi_{\beta_1,\beta_2}(\omega),\\
1 & $ if$ \quad  y >\psi_{\beta_1,\beta_2}(\omega)
\end{array}\right.
\end{align}
where $g_{\omega,\beta_1,\beta_2}^n(y)$ is given by (\ref{com}). To prove the theorem, we closely follow \cite[Theorem 3.1]{GH} and omit some details.
Let $H_{\beta_1,\beta_2}^+ $ be the skew product map whose fiber maps are defined by $h_{i,\beta_i}=g_{i,\beta_i}^{-1}$. Then, it has positive normal Lyapunov exponents along $\mathbb{A}_0$ and $\mathbb{A}_1$.
This means that $L_{\beta_1,\beta_2}(0)>0$ and $L_{\beta_1,\beta_2}(1)>0$ for the skew product map $H_{\beta_1,\beta_2}^+ $.

\begin{lemma}\label{99}
For the skew product $H_{\beta_1,\beta_2}^+ $ defined above, there exists
an ergodic stationary measure $m$ with $m(\{0\} \cup \{1\}) = 0$.
\end{lemma}
\begin{proof}
The lemma is just a reformulation of \cite[Lemma 3.2]{GH} in our context. Note that to prove \cite[Lemma 3.2]{GH} the authors use only the fact that
the normal Lyapunov exponents at both fixed points $0$ and $1$ are positive. So the result holds for our setting.
\end{proof}
Take the extension skew product $H_{\beta_1,\beta_2}$ of $H_{\beta_1,\beta_2}^+$ given by (\ref{s1}).
Note that the stationary measure $m$ gives an invariant measure $\mu_m$ for the extension skew product
system $H_{\beta_1,\beta_2}$ (see \cite{Ar}). Its conditional measures on fibers $\{\omega\} \times \mathbb{I}$
are denoted by $\mu_{m,\omega}$ (see (\ref{co1}) and (\ref{co2})).
By \cite[Lemma 3.3]{GH}\label{lem1}, the conditional
measure $\mu_{m,\omega}$ of $\mu_m$ is a $\delta$-measure for $\mathbb{P}$-almost every $\omega \in \Sigma_2$.

By \cite[Lemma 3.4]{GH}, the stationary measure $m$ obtained by Lemma \ref{99} is unique. Thus, for each $(\beta_1,\beta_2) \in \Gamma_S$, there is a unique stationary measure $m_{\beta_1,\beta_2}$
 with $m_{\beta_1,\beta_2}(\{0\} \cup \{1\}) = 0$.
So, by these facts, there exists a measurable function $\psi_{\beta_1,\beta_2} : \Sigma_2 \to \mathbb{I}$ such that for $\mathbb{P}$-almost all $\omega$,
\begin{equation*}
 \lim_{n \to \infty} h^n_{\sigma^{-n}\omega,\beta_1,\beta_2} m_{\beta_1,\beta_2}=\delta_{\psi_{\beta_1,\beta_2}(\omega)},
\end{equation*}
where $h^n_{\sigma^{-n}\omega,\beta_1,\beta_2}$ is given by (\ref{com}).
Note that a function increasing if and only if its inverse is increasing. As the convex hull of the support of $m_{\beta_1,\beta_2}$ equals $\mathbb{I}$ and
 $h_{1,\beta_1}$, $h_{2,\beta_2}$ are increasing, this implies that for every $y \in (0, 1)$
\begin{equation*}
 \lim_{n \to \infty}h^n_{\sigma^{-n}\omega,\beta_1,\beta_2}(y)=\psi_{\beta_1,\beta_2}(\omega).
\end{equation*}
By the fact that $h_{1,\beta_1}$, $h_{2,\beta_2}$ are increasing,
we get
\begin{equation*}
 \lim_{n \to \infty}(h^n_{\sigma^{-n}\omega,\beta_1,\beta_2})^{-1}(y)=1
\end{equation*}
if $y> \psi_{\beta_1,\beta_2}(\omega)$ and
\begin{equation*}
 \lim_{n \to \infty}(h^n_{\sigma^{-n}\omega,\beta_1,\beta_2})^{-1}(y)=0
\end{equation*}
if $y < \psi_{\beta_1,\beta_2}(\omega)$.
Thus,
\begin{equation*}
 \lim_{n \to \infty}g^n_{\omega,\beta_1,\beta_2}(y)=1
\end{equation*}
if $y> \psi_{\beta_1,\beta_2}(\omega)$ and
\begin{equation*}
 \lim_{n \to \infty}g^n_{\omega,\beta_1,\beta_2}(y)=0
\end{equation*}
if $y < \psi_{\beta_1,\beta_2}(\omega)$.

This observation shows that the union
of the basins of attraction of $\mathbb{A}_0$ and $\mathbb{A}_1$ has full standard measure and hence, Theorem \ref{t1} is proved.
\end{proof}
The following result describes intermingled basins for step skew product systems $G^+_{\beta_1,\beta_2} \in \mathcal{S}$.
\begin{theorem}\label{t2}
Let $G^+_{\beta_1,\beta_2} \in \mathcal{S}$ be a step skew product whose fiber maps defined by (\ref{fiber1}) and $(\beta_1,\beta_2)\in \Gamma_S$, where $\Gamma_S$ is given by (\ref{region1}).
Let $\beta_1=\beta_2=\beta$. Then, the invariant subset $\mathbb{A}_\beta=\Sigma_2^+ \times \{\beta\}$ is a Milnor attractor
 and the basins of $\mathbb{A}_i$ and $\mathbb{A}_\beta$ are intermingled, for each $i=0,1$.
\end{theorem}
\begin{proof}
Let $(\beta_1,\beta_2)\in \Gamma_S$ and $\beta_1=\beta_2=\beta$. Then, the subset $\mathbb{A}_\beta=\Sigma_2^+ \times \{\beta\}$ is invariant. By (\ref{FL}) and the definition of $\Gamma_S$,
the normal Lyapunov exponents at the points 0, 1 and $\beta$ are negative.
Thus, by Theorem \ref{t1}, $\mathbb{A}_\beta$ is a Milnor attractor.
Now, we take the restriction skew products
\begin{equation}\label{s4}
G^+_{\beta_1,\beta_2}|_{\Sigma_2 \times [0,\beta]} \quad \text{and} \quad G^+_{\beta_1,\beta_2}|_{\Sigma_2 \times [\beta,1]}.
\end{equation}

The restriction skew products given by (\ref{s4}) satisfy the hypothesis of \cite[Theorem 3.1]{GH}, thus, the basins of $\mathbb{A}_i$ and $\mathbb{A}_\beta$ are intermingled, for each $i=0,1$.
This finishes the proof.
\end{proof}
Let us define $\tau:\mathbb{I} \to \{1,2\}$ by
\begin{align}\label{code}
\tau(x) =\left\{\begin{array}{cc}
1 & $ for$ \quad0\leq x\leq 1/2,\\
2 & $ for$ \quad 1/2< x\leq1.
\end{array}\right.
\end{align}
Let $\Sigma_{12}\subset \Sigma_2$ be the sets of bi-infinite sequences of 1's and 2's which do not end with a tail
of 1s or 2s.
The metric and
measure is inherited from the space $\Sigma_2$.
Note that $\Sigma_{12}$ is closed and invariant by the shift map.
By \cite{Go}, there exists a map $K$ that semi-conjugates
the restrictions $\sigma |_{\Sigma_{12}}$ and $f|_{\mathbb{I} \setminus D}$,
where $D$ is the set of dyadic rationals (i.e. rational numbers of the form $k/2^n$, whose denominator is a power of 2).
By definition, the subset $\Sigma_{12}$ has the full Bernoulli measure and $D$ has zero Lebesgue measure (see \cite{Go}).
Then, the composition of the semi-conjugating map $K$ and the map $\psi_{\beta_1,\beta_2}$ given by (\ref{xi})
yields a measurable map $\Phi:\mathbb{I} \to \mathbb{I}$ that separates the basins.
\begin{corollary}\label{gg}
There exists a measurable graph map $\Phi_{\beta_1,\beta_2}:\mathbb{I} \to \mathbb{I}$ such that for Lebesgue-almost all $x \in \mathbb{I}$,
\begin{align}\label{phi}
\lim_{n \to \infty}g_{x,\beta_1,\beta_2}^n(y) =\left\{\begin{array}{cc}
0 & $ if$ \quad y< \Phi_{\beta_1,\beta_2}(x),\\
1 & $ if$ \quad  y >\Phi_{\beta_1,\beta_2}(x)
\end{array}\right.
\end{align}
where $g_{x,\beta_1,\beta_2}^n(y)=g_{\tau(f^{n-1}(x)),\beta_{\tau(f^{n-1}(x))}} \circ \dots \circ g_{\tau(x),\beta_{\tau(x)}}(y)$.
\end{corollary}
The following is immediate from Theorem \ref{t1} and Corollary \ref{gg}.
\begin{corollary}\label{ggg}
Let $F_{\beta_1,\beta_2}\in \mathcal{F}$ be a skew product of the form (\ref{ss}) whose fiber maps $g_{i,\beta_1,\beta_2}$, $i=1,2$, given by (\ref{fiber1}).
If both normal Lyapunov exponents $L_{\bot,\beta_1,\beta_2}(i)$, $i=0,1$, are negative,then
the attractors $A_0$ and $A_1$ attract sets of positive Lebesgue measure and the union of the basins $\mathcal{B}_{\beta_1,\beta_2}(A_i)$, $i=0,1$, has full Lebesgue measure.
\end{corollary}

\section{\textbf{The existence of riddled basin and blowout bifurcation}}
In this section, we demonstrate the emergence conditions of riddled basins in the global sense.\\
Note that a general set of conditions under which riddled basins
can occur are as follows (see \cite{OS}, \cite{JI}, \cite{C} and \cite{RS}):
\begin{enumerate}
\item [$(H1)$] There exists an invariant subspace $N$ whose dimension is less than the dimension of the full-phase space.
\item [$(H2)$] The dynamics on the invariant subspace $N$ has a chaotic attractor $A$.
\item [$(H3)$] For typical orbits on $A$ the Lyapunov exponents for infinitesimal perturbations in the direction transverse to
 $N$ are all negative.
 \item [$(H4)$] There is another attractor $A^{\prime}$ not belonging to $N$.
\item [$(H5)$] At least one of the normal Lyapunov exponents, although negative for almost any orbits on $A$, has finite time fluctuations that are positive.
\end{enumerate}
\begin{remark}
Note that for occurrence the riddled basin, it is necessary to have a dense set of points with
zero Lebesgue measure in the attractor lying in the invariant subspace
which are transversely unstable, thus it is necessary that this
attractor be chaotic.
\end{remark}
Let $F_{\beta_1,\beta_2}\in \mathcal{F}$ be a skew product of the form (\ref{ss}) whose fiber maps $g_{i,\beta_1,\beta_2}$, $i=1,2$, given by (\ref{fiber1}) and $\beta_1, \beta_2 \in (0,1)$.
In Theorem \ref{th1}, we estimated the values of parameters $\beta_i$, $i=1,2$, for which the locally riddled basin occurs for both attractors $A_0$ and $A_1$.
In particular, for each $(\beta_1,\beta_2) \in \Gamma_S$, the verification of condition $(H1)-(H4)$ was done.
Condition $(H5)$ is verified in the next theorem.
To do that we show the existence a set of unstable periodic orbits embedded in $A_i$ which is transversely unstable.
This implies that at least one of the Lyapunov exponents along the directions transverse to invariant subspaces experiences positive finite-type fluctuations (see \cite{RS}).

Here, we prove the riddling in a more direct way.
Indeed, in our setting, the full space contains two chaotic attractors $A_i$, $i=0,1$, lying in different invariant subspaces $N_i$. By Corollary \ref{gg}, the system $F_{\beta_1,\beta_2}$
presents a complex fractal boundary between the initial conditions leading to each of the two attractors.
This fractal boundary is the graph of $\Phi_{\beta_1,\beta_2}$ which separates the basin of attraction.
Note that in a riddled basin, small variations in initial conditions induce a switch between the different
chaotic attractors but the fractal boundary causes to predict, from a given initial condition, what trajectory
in phase-space the system will follow.
These facts allow us to verify the occurrence of riddled basin (in the global sense).
\begin{theorem}\label{m}
Let $F_{\beta_1,\beta_2}\in \mathcal{F}$ be a skew product of the form (\ref{ss}) whose fiber maps $g_{i,\beta_1,\beta_2}$, $i=1,2$, given by (\ref{fiber1}) and let $(\beta_1, \beta_2) \in \Gamma_S$, where
$\Gamma_S$ is given by (\ref{region1}).
Consider the chaotic attractors $A_{i}$, $i=0,1$. Then the following holds:
\begin{itemize}
\item[$(a)$] if $\beta_1 \leq1/2$ and $ \beta_2 <\frac{1}{1-\beta_1}-1$ then $\mathcal{B}(A_{0})$ is riddled with $\mathcal{B}(A_{1})$;
\item[$(b)$]   if $\beta_1 \geq 1/2$ and $\beta_2 > 2 - \frac{1}{\beta_1}$, then $\mathcal{B}(A_{1})$ is riddled with $\mathcal{B}(A_{0})$;
\end{itemize}
where $\mathcal{B}(A_{i})$ is the basin of attraction of $A_i$, for $i=0,1$.
\end{theorem}
\begin{proof}
As we have seen in Section 3, the two chaotic attractors $A_i$, $i=0,1$, are $SRB$ attractors for the restriction of $F_{\beta_1,\beta_2}$ to the invariant subspaces $N_i$.
Since $(\beta_1, \beta_2) \in \Gamma_S$, by Corollary \ref{cor1},  both normal Lyapunov exponents $L_{\perp, \beta_1,\beta_2}(0)$
 and $L_{\perp, \beta_1,\beta_2}(1)$ are negative.  This fact ensures that $A_i$, $i=0,1$, are (essential) attractors (in the Milnor sense) in the whole phase space (see Theorem \ref{th1}).

To prove $(a)$,
consider the fixed point $Q=(1,0)\in A_{0}$.
We closely follow \cite[Section~4]{JI} and  construct an open set near the fixed point $Q$ which is not contained in the basin $\mathcal{B}(A_{0})$.
Indeed, consider the graph $x\mapsto (1-x)^{ \log \frac{(1+\beta_2)}{2}}$, a subset of $A_{0}\times[0,1]$ having a cusp at $Q=(1,0)$. This graph is strictly monotone and concave on each side of its cusp point.
Since the mapping $y \mapsto dg_{2,\beta_2}(y)$ is continuous, there exists a real $ \gamma(\beta_1,\beta_2)>0$ sufficiently close to 1 such that
if we take
\begin{equation*}
W_{\beta_1,\beta_2}^{+}=\{(x,y)\in [0,1]\times[0,1] :\gamma(\beta_1,\beta_2) <x<1 \ \text{and} \ y>(1-x)^{ \log (1+\beta_2)}\}
  \end{equation*}
then for $(x,y)\in W_{\beta_1,\beta_2}^{+}$, $dg_{2,\beta_2}(y)>1$. Thus any point in $W_{\beta_1,\beta_2}^{+}$ escapes from the attractor $A_{0}$.
By Corollary \ref{ggg}, the union of the basins $\mathcal{B}_{\beta_1,\beta_2}(A_i)$, $i=0,1$, has full Lebesgue measure.
Also, the graph of the map $\Phi_{\beta_1,\beta_2}$ obtained in Corollary \ref{gg} is the fractal boundary which separates the basin of attraction.
By these facts, if $ \gamma(\beta_1,\beta_2)$ is close enough to 1, for any point $(x,y) \in W_{\beta_1,\beta_2}^{+}$, there is $n>1$ such that $g_{x,\beta_1,\beta_2}^n(y)>\Phi_{\beta_1,\beta_2}(x)$
and hence $(x,y)$ escapes to $\mathcal{B}(A_1)$ after some iterates.

Let $W_{\beta_1,\beta_2}=\cup_{n\geq0}F_{\beta_1,\beta_2}^{-n}(W_{\beta_1,\beta_2}^{+})$. Then $W_{\beta_1,\beta_2}$ is an open set and its boundary has a cusp at each point of the set $\{f^{-n}(1)\}$. The set $\{f^{-n}(1)\}$ is dense in $\mathbb{I}$  and thus every neighborhood of any point in $A_{0}$ intersects $W_{\beta_1,\beta_2}$. Thus the basin $\mathcal{B}(A_{0})$ is riddled with the basin $\mathcal{B}(A_{1})$.

To prove ($b$), we construct an open set near the fixed point $Z=(1,1)\in A_{1}$ which is not contained in the basin $\mathcal{B}(A_{1})$.
Consider the graph $x\mapsto x^{\log(2-\beta_2)}$, a subset of $A_{1}\times[0,1]$ which has a cusp at $Z=(1,1)$. This graph is strictly monotone and concave on each side of its cusp point.
Since the mapping $y \mapsto dg_{2,\beta_2}(y)$ is continuous, there exists a real $ \eta(\beta_1,\beta_2)>0$ sufficiently close to 0 such that
if we take
 \begin{equation*}
U_{\beta_1,\beta_2}^{+}=\{(x,y)\in [0,1]\times[0,1] : 0<x<\eta(\beta_1,\beta_2) \ \text{and} \ y> x^{log(2-\beta_2)}\},
\end{equation*}
then for $(x,y)\in U^{+}$, $dg_{2,\beta_2}(y)>1$. Thus any point in $U_{\beta_1,\beta_2}^{+}$ escapes from the attractor $A_{1}$.
By applying an argument similar to statement ($a$), any point $(x,y) \in U_{\beta_1,\beta_2}^{+}$ escapes to $\mathcal{B}(A_0)$ after some iterates.
Let $U=\cup_{n\geq0}F^{-n}(U^{+})$. Then $U$ is an open set and its boundary has a cusp at each point of the set $\{f^{-n}(1)\}$. The set $\{f^{-n}(1)\}$ is dense in $[0,1]$ and thus every neighborhood of any point in $A_{1}$ intersects $U$.
Thus, $\mathcal{B}(A_{1})$ is riddled with the basin $\mathcal{B}(A_{0})$.
 \end{proof}

\begin{figure}[h!]
\begin{center}
    $\begin{array}{cc}
  \includegraphics[scale=0.4]{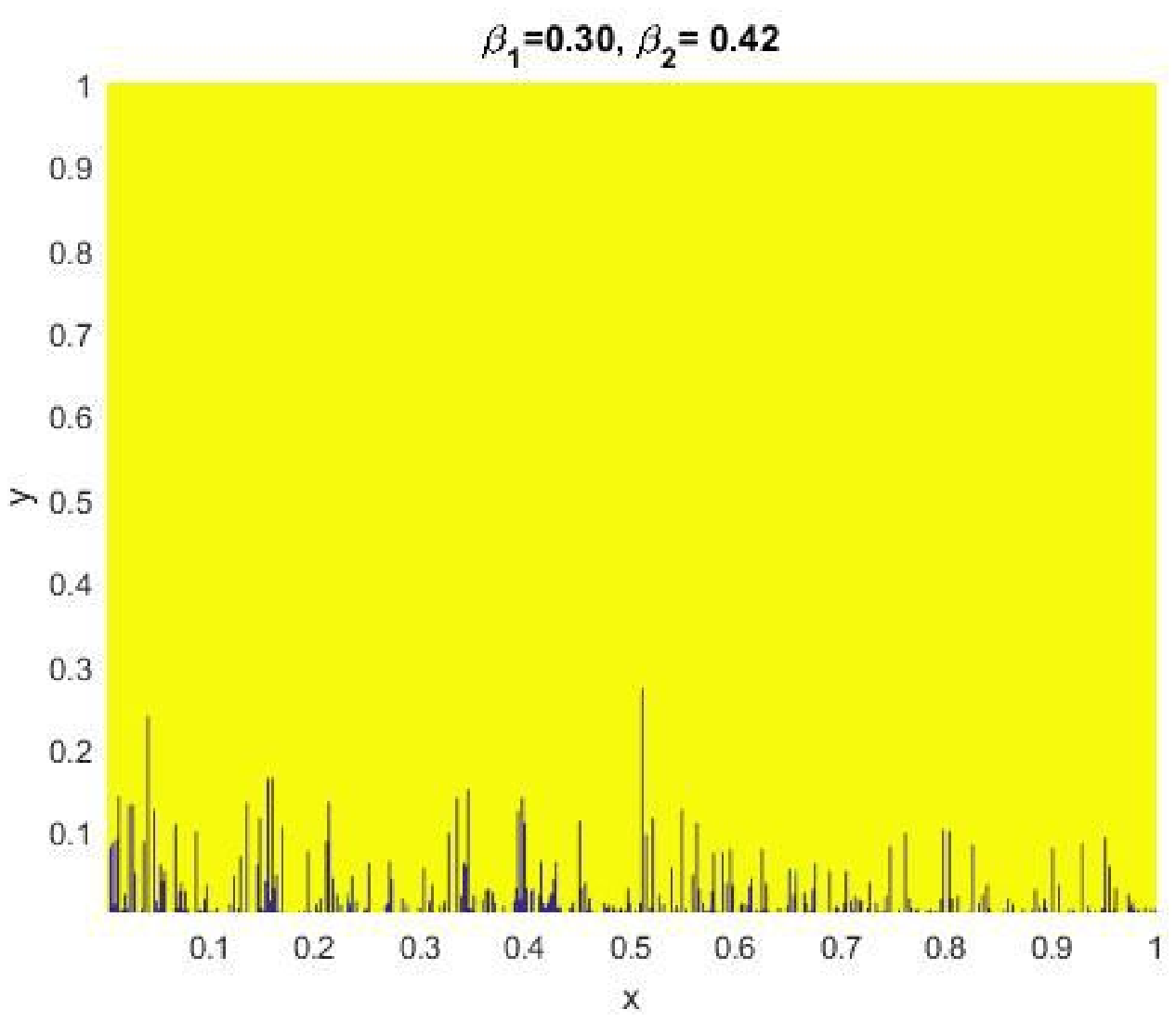}
 \includegraphics[scale=0.4]{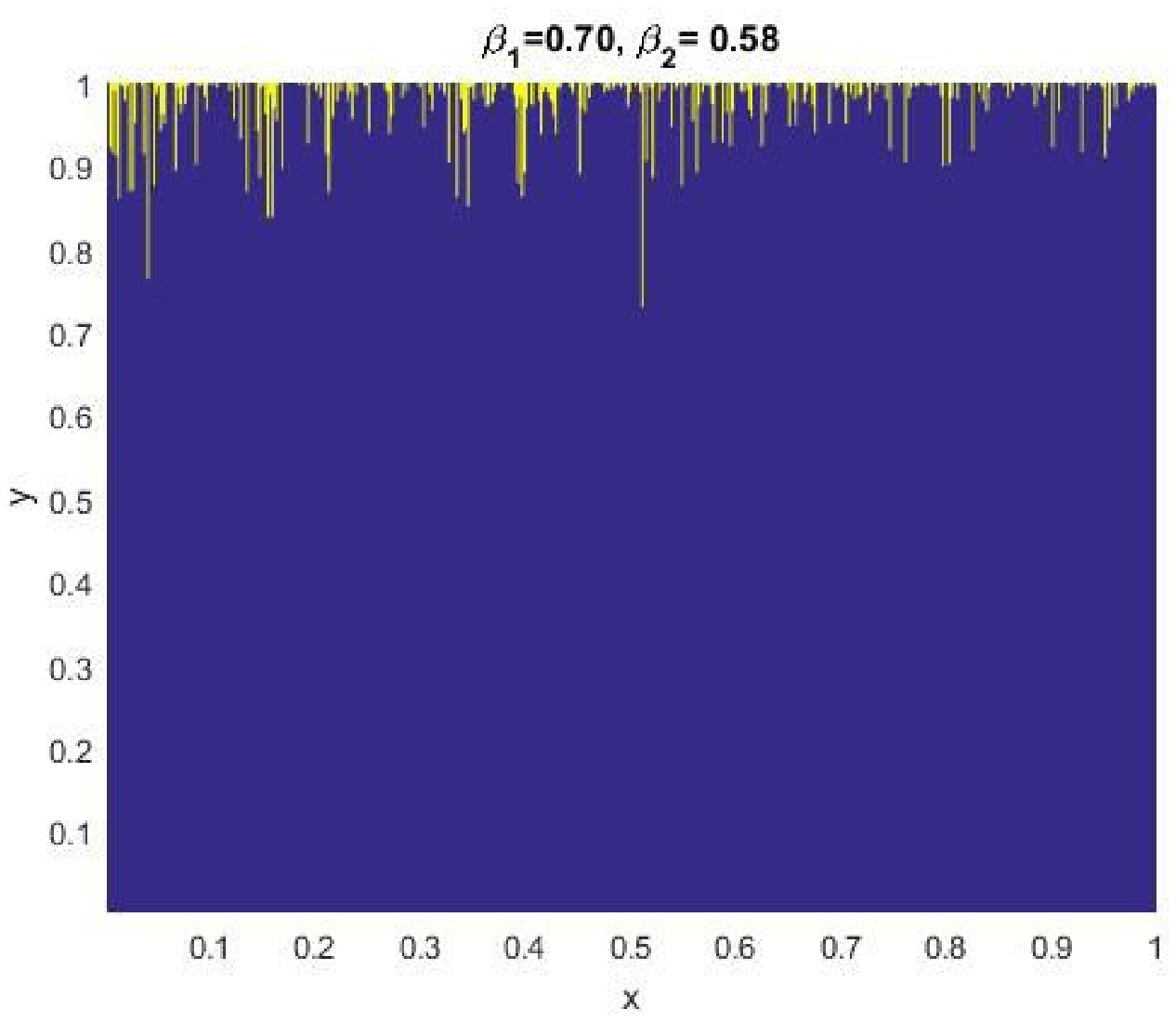}
  \end{array}$
\end{center}
  \caption{\small The left frame depicts the basins of attraction of $A_0$ and $A_1$, for $\beta_1=.3$ and $\beta_2=.42$. The basin of $A_0$ is riddled by the basin of $A_1$.
The right frame shows the basins of attraction of $A_0$ and $A_1$
 for $\beta_1=.7$ and $\beta_2=.58$. The basin of $A_1$ is riddled by the basin of $A_0$. The blue region in both figures corresponds to the basin of attraction $\mathcal{B}(A_0)$, while
 the yellow region corresponds to the basin of attraction $\mathcal{B}(A_1)$. }
 \label{fig:5}
   \end{figure}

Let $(\beta_1, \beta_2) \in \Gamma_S$. Take $\beta:=\beta_1=\beta_2$ and
\begin{equation}
 N_{\beta}:=\{(x,y) : \quad 0\leq x\leq 1, \quad y=\beta\}.
\end{equation}
Then $N_{\beta}$ is an invariant subspace by $F_{\beta,\beta}$ and the restriction of $F_{\beta,\beta}$ to this invariant subspace possesses a chaotic attractor $A_{\beta}$.

The normal Lyapunov exponent for $F_{\beta,\beta}$ is given by
 \begin{align}
 L_{\perp,\beta,\beta}(y)=\frac{1}{2} \ln(-\beta^2+\beta+1)+\frac{1}{2} \ln(\beta^2-\beta+1).
 \end{align}
Simple computation shows that $ L_{\perp,\beta,\beta}$ is negative for each $\beta \in (0,1)$. Note that the intervals $[0,\beta]$ and $[\beta,1]$ are invariants by both fiber maps $g_{1,\beta}$ and $g_{2,\beta}$.
By these facts, Theorem \ref{t1} and the comments before Corollary \ref{gg}, we get the next result.
\begin{corollary}
$F_{\beta,\beta}$ exhibit three Milnor attractors $A_0$, $A_1$ and $A_\beta$. Moreover, the basins of $A_i$ and $A_\beta$ are intermingled, for each $i = 0, 1$.
\end{corollary}

\begin{figure}[h!]
\begin{center}
    $\begin{array}{c}
  \includegraphics[scale=0.5]{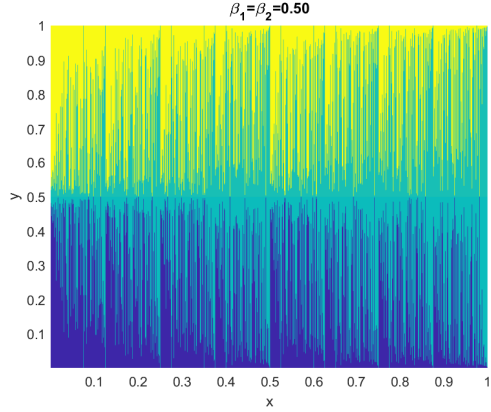}
  \end{array}$
\end{center}
  \caption{\small The basins of attraction for the attractors $A_0$, $A_1$ and $A_\beta$, for parameters values $\beta_1=\beta_2=\beta=0.5$. The blue region in the figure corresponds to the basin of attraction $\mathcal{B}(A_0)$,
  the green region in the figure corresponds to the basin of attraction $\mathcal{B}(A_{\beta})$,
  while yellow region corresponds to the basin of attraction $\mathcal{B}(A_1)$. The intermingled basin is observed.}
 \label{fig:4}
   \end{figure}

We recall that the blowout bifurcation \cite{EJ, PJ} occurs when $L_{\perp,\beta_1,\beta_2}$ crosses zero.
\begin{corollary}
Let $F_{\beta_1,\beta_2}\in \mathcal{F}$ be a skew product of the form (\ref{ss}) whose fiber maps $g_{i,\beta_1,\beta_2}$, $i=1,2$, given by (\ref{fiber1}).
 Then the following holds:
\begin{itemize}
\item[$(a)$] $F_{\beta_1,\beta_2}$ exhibits a (subcritical) hysteretic blowout bifurcation on passing
through any $\beta_1 <1/2$ and $ \beta_2 =\frac{1}{1-\beta_1}-1$;
\item[$(b)$]  $F_{\beta_1,\beta_2}$ exhibits a (subcritical) hysteretic blowout bifurcation on passing
through any $\beta_1 >1/2$ and $\beta_2 = 2 - \frac{1}{\beta_1}$.
\end{itemize}
\end{corollary}
\begin{proof}
By the proof of Theorem \ref{th1}, for each $(\beta_1,\beta_2)\in \Gamma_S$, the normal Lyapunov exponents $L_{\perp,\beta_1,\beta_2}(i)$, $i=0,1$, is smoothly dependent
on the normal parameters $\beta_1$ and $\beta_2$. By this fact and Theorem \ref{th1}, $L_{\perp,\beta_1,\beta_2}(i)$ crosses zero at any $\beta_1 <1/2$ and $ \beta_2 =\frac{1}{1-\beta_1}-1$
and at any $\beta_1 >1/2$ and $\beta_2 = 2 - \frac{1}{\beta_1}$. Thus, by definition of a hysteretic blowout bifurcation and Theorem \ref{th1}, the result is immediate.
\end{proof}
\section{\textbf{Conclusion}}
We have investigated the formation of locally riddled basins of attraction and chaotic saddle for a two parameter family $F_{\beta_1,\beta_2}$, $\beta_1,\beta_2 \in (0,1)$, of skew product systems defined on the plane.
Our model exhibits two distinct chaotic attractors $A_0$ and $A_1$ lying in two different invariant subspaces.
We have analyzed the model rigorously using the results of \cite{PJ}
and estimated the range of values of parameters $\beta_i$, $i=1,2$,
 such that the attractor $A_0$ or $A_1$ has a locally riddled basin, or becomes
a chaotic saddle.
Then by varying the parameters $\beta_i$, $i=1,2$, in an open region in the $\beta_1\beta_2$-plane, we have shown the occurrence of riddled basin (in the global sense) and hysteretic blowout bifurcation.
To prove the riddling basin, we have semi-conjugated the system to a random walk model and provided a complex fractal boundary between
the initial conditions leading to each of the two attractors. This boundary separates the basin of attraction and causes to predict, from a
given initial condition, what trajectory in phase-space the system will follow.
Moreover, it was shown that, by varying the parameters in an open region, a new chaotic attractor appears when $\beta_1=\beta_2$.
Also, the basin of this new attractor intermingled with the basins of both attractors $A_0$ and $A_1$.
Numerical simulations were presented graphically
to confirm the validity of our results.




\end{document}